\documentclass{article}
    
\usepackage{subfiles}
\usepackage{titling}  

\usepackage[utf8]{inputenc}
\usepackage[margin=1.5in]{geometry}
\usepackage{setspace}
\usepackage{amsmath,amsfonts,amstext}
\usepackage{amsthm}
\usepackage{graphicx}
\usepackage{booktabs}
\usepackage{multirow}
\usepackage{graphicx}
\usepackage[shortlabels]{enumitem}
\usepackage[hidelinks]{hyperref}
\usepackage[all]{hypcap}
\usepackage[numbers,sort&compress]{natbib}
\newcommand\Tstrut{\rule{0pt}{3ex}}         
\newcommand\Bstrut{\rule[-1.8ex]{0pt}{0pt}}   

\newcommand{\abs}[1]{\left\lvert #1 \right\rvert}
\newcommand{\N}{\mathbb{N}}
\newcommand{\RO}{\mathcal{R}_0}
\renewcommand{\P}{\mathbb{P}}

\renewcommand{\phi}{\varphi}
\renewcommand{\epsilon}{\varepsilon}
\DeclareRobustCommand{\stirling}{\genfrac\{\}{0pt}{}}
\DeclareMathOperator{\Pois}{Pois}
\DeclareMathOperator{\NB}{NB}

\newtheorem{theorem}{Theorem}

\newtheorem{lemma}[theorem]{Lemma}
\newtheorem{corollary}[theorem]{Corollary}

\newtheorem*{examples*}{Examples}
\newtheorem*{example*}{Example}

\newtheorem{innercustomgeneric}{\customgenericname}
\providecommand{\customgenericname}{}
\newcommand{\newcustomtheorem}[2]{%
  \newenvironment{#1}[1]
  {%
   \ifdefined\crefalias\crefalias{innercustomgeneric}{#2}\fi
   \renewcommand\customgenericname{#2}%
   \renewcommand\theinnercustomgeneric{##1}%
   \innercustomgeneric
  }
  {\endinnercustomgeneric}%
  \ifdefined\crefname\crefname{#2}{#2}{#2s}\fi
}

\newcustomtheorem{customtheorem}{Theorem}
\newcustomtheorem{customlemma}{Lemma}

\theoremstyle{definition}

\newtheorem*{definition*}{Definition}

\newtheorem*{remark*}{Remark}

\begin{document}

\title{
\vspace{2cm}
Tracking dynamics of superspreading through contacts, exposures, and transmissions in edge-based network epidemics
\vspace{1cm}}
\author{
Ari S. Freedman$^{1,2,3,*}$, Bjarke F. Nielsen$^{4,5,6}$, Maximillian M. Nguyen$^{7,8}$,\\Laurent H\'ebert-Dufresne$^{3,9,10}$, Simon A. Levin$^1$}
\date{
\small
$^1$Department of Ecology and Evolutionary Biology, Princeton University, Princeton, NJ, USA\\
$^2$Department of Plant Biology, University of Vermont, Burlington, VT, USA\\
$^3$Vermont Complex Systems Institute, University of Vermont, Burlington, VT, USA\\
$^4$High Meadows Environmental Institute, Princeton University, Princeton, NJ, USA\\
$^5$Niels Bohr Institute, University of Copenhagen, Copenhagen, Denmark\\
$^6$PandemiX Center, Roskilde University, Roskilde, Denmark\\
$^7$Lewis-Sigler Institute, Princeton University, Princeton, NJ, USA\\
$^8$School of Medicine, Emory University, Atlanta, GA, USA\\
$^9$Complexity Science Hub, Vienna, Austria\\
$^{10}$Santa Fe Institute, Santa Fe, NM, USA\\
$^*$Corresponding author; email: ari.freedman@uvm.edu
}

\onehalfspacing

\emergencystretch 3em

\maketitle

\vspace{3cm}

Keywords: superspreading, networks, disease modeling, dispersion, SARS-CoV-2

\newpage

\section*{Abstract}

Infectious disease superspreading caused by heterogeneity in contact behavior has been observed to be an important determinant of epidemic dynamics and size in both empirical and theoretical settings. However, it has also been observed that the importance of this type of superspreading changes throughout an epidemic, generally in a decreasing manner as infections cascade from individuals with many contacts to those with fewer contacts. We provide an exact mathematical formulation of this phenomenon in strongly-immunizing (SIR) epidemics on static contact networks. Building on the edge-based modeling framework, we construct three metrics to track how superspreading changes through the course of an epidemic, respectively measuring infected nodes' contacts, exposures, and transmissions: (1) the mean degree of infected nodes, (2) the mean number of susceptible neighbors of infected nodes, and (3) the mean number of secondary cases that will be caused by newly infected nodes.
We prove results about the behaviors of these metrics, highlighting the fact that their peak times all occur at less than half the time it takes for population-level infection prevalence to peak. This suggests that the importance of superspreading will be low when an epidemic is already near its peak, so contact-based control strategies are best employed as early in an outbreak as possible. We discuss implications for accurately measuring epidemiological parameters from incidence, mobility, contact tracing, and transmission data.

\section*{Introduction}
Infectious disease models traditionally assume individuals in a population are well-mixed in their contact patterns, assuming mass-action transmissions where incidence is proportional to both the number of infected individuals and the number of susceptible individuals \cite{KermackMcKendrick1927,Heesterbeek2005}. However, human populations are highly structured in their contact patterns, a phenomenon which models often capture abstractly using contact networks \cite{Klovdahl1985,WattsStrogatz1998,Newman2002,EubankEtAl2004,BalcanEtAl2009}. Much work has been done on mathematically describing the spread of infectious disease and other spreading process on networks using a variety of methods, including branching process \cite{LloydSmithEtAl2005}, moment closure \cite{EamesKeeling2002,Bauch2002}, degree block \cite{BarratEtAl2013}, approximate master equation \cite{MarceauEtAl2010, HebertDufresneEtAl2010,Gleeson2013}, effective degree \cite{BallNeal2008,LindquistEtAl2011}, message passing \cite{KarrerNewman2010,SherborneEtAl2018}, generation-based \cite{NoelEtAl2009} and edge-based approaches \cite{Volz2008,Miller2011}. Particular emphasis has been placed on models with perfect immunity, adapting the original SIR model of Kermack and McKendrick \cite{KermackMcKendrick1927} to a static contact network with an arbitrary degree distribution. Of these, Volz \cite{Volz2008} achieved the first low-dimensional mathematical description of a SIR epidemic evolving over a network in continuous time, an approach further simplified by Miller \cite{Miller2011} down to a system of two ordinary differential equations. 

A motivating force behind this proliferation of network infectious disease models has been the increasingly recognized significance of superspreading in epidemics \cite{Klovdahl1985,LloydSmithEtAl2005,BrainardEtAl2023}, where a small group of individuals are responsible for a majority of transmissions (during the SARS-CoV-2 pandemic, 10\% of cases caused as much as 80\% of infections \cite{MillerEtAl2020, EndoEtAl2020, LauEtAl2020, AlthouseEtAl2020}). Superspreading events can be driven by various factors, from heterogeneity in biological factors like viral shedding volume to behavioral factors such as the size of one's social contact network \cite{NielsenEtAl2023, GoyalEtAl2021, AlthouseEtAl2020}. Here, we focus on contact heterogeneity, which has been shown to have important effects on epidemic dynamics \cite{LloydSmithEtAl2005,Volz2008,MillerEtAl2012}, critical outbreak threshold \cite{Andersson1997,MayLloyd2001,Newman2002,PastorSatorrasVespignani2002,MillerEtAl2012}, final size \cite{Miller2011,GrossmannEtAl2021,NoelEtAl2009}, herd immunity threshold \cite{BrittonEtAl2020,GomesEtAl2022,OzEtAl2021}, extinction probabilities \cite{LloydSmithEtAl2005,DiekmannEtAl2012}, evolutionary potential \cite{LeventhalEtAl2015}, and effectiveness of control efforts \cite{LloydSmithEtAl2005,KissEtAl2005,NielsenEtAl2023}.

Another previously described phenomenon of epidemic spreading on heterogeneous networks is the propensity for high degree nodes (representing individuals with many contacts) to be more likely to become infected and to do so earlier, as they are the ones most likely to be connected by a random edge (or contact) in the network \cite{AllardEtAl2023,BarthelemyEtAl2005}.
This is the same logic underlying the ``friendship paradox'', that a random friend of an individual will on average have more friends than that individual \cite{AllardEtAl2023}. Furthermore, if the infection has long-lasting immunity, then the infection will tend to cascade from high degree nodes to low degree nodes as the high degree nodes become infected first but then also recover and gain immunity first \cite{BarthelemyEtAl2005}. Thus, the potential for superspreading in an epidemic may decline over time. Statistical approaches have been developed to estimate the varying role of superspreading during the SARS-CoV-2 pandemic \cite{MiyamaEtAl2022,GuoEtAl2023}. However, the literature has lacked a rigorous analytic exploration of how superspreading potential changes over time in a network epidemic by tracking the degrees of infected nodes over time.

In this work, we define three metrics of superspreading potential, measuring (1) the average number of contacts that infected nodes have, (2) the average number of those contacts which are susceptible, and (3) the average number of those susceptible contacts to which the infected node transmits the contagion. These metrics capture how the numbers of contacts, exposures, and transmissions associated with infected nodes change over time. We derive insights into how these metrics evolve in a continuous-time SIR epidemic over a static contact network with given degree distribution $K$, building on Miller and Volz's edge-based model \cite{Miller2011,MillerEtAl2012}.

Lastly, we discuss how these metrics are involved in inference methods for key epidemiological parameters from various sources of epidemic data. Importantly, our results show that inferred epidemiological parameters may differ greatly depending on the type of data available and how they are temporally aggregated.

\subsection*{Overview of the superspreading metrics and analytic results}

We formally define three metrics of superspreading potential, each of which is the mean of a distribution related to the properties of infected nodes at a given time. First, we define the \textit{infected degree distribution} $X(t)$ as the degree distribution of infected nodes at time $t$, which has mean $m(t)$, variance $v(t)$, moments $m_n(t)$, and mass function $p_k(t)$. Assuming that the initially seeded infected nodes are randomly chosen, the infected degree distribution will start out the same as the overall degree distribution of the network, $X(0)=K$ and $m(0)=\mu$. As we will show, $m(t)$ first increases as $X(t)$ approaches the ``neighbor degree distribution'' $K_{\textnormal{n}}$ (degree distribution of random neighbors) due to high degree nodes getting primarily infected first, after which $m(t)$ declines as high degree nodes recover and the infection moves to low degree nodes.

\begin{table}[!b]
\centering{
\begin{tabular}{cccccc}
\hline
Distribution & rv & pmf & mean & variance & $n$-th moment\\
\hline
Degree & $K$ & $P(k)$ & $\mu$ & $\nu$ & $\mu_n$\\
Infected degree & $X(t)$ & $p_k(t)$ & $m(t)$ & $v(t)$ & $m_n(t)$\\
Effective degree & $E(t)$ & $p_{E,k}(t)$ & $m_E(t)$ & $v_E(t)$ & $m_{E,n}(t)$\\
Secondary case & $Z(t)$ & $p_{Z,k}(t)$ & $m_Z(t)$ & $v_Z(t)$ & $m_{Z,n}(t)$\\
\hline
\end{tabular}}
\caption{Notation for the random variables (rv), probability mass functions (pmf), means, variances and moments associated with network degree, infected degree, effective degree, and secondary case distributions.}
\label{table:notation}
\end{table}

Second, we define the \textit{effective degree distribution} $E(t)$ as the distribution of the number of susceptible neighbors each infected node has at time $t$ (named following Ref. \cite{LindquistEtAl2011}), which has mean $m_E(t)$, variance $v_E(t)$, moments $m_{E,n}(t)$, and mass function $p_{E,k}(t)$. This distribution captures the idea that while an infected node may have a large number of neighbors, its ability to transmit to many neighbors depends on it having a large number of susceptible neighbors, as these are the ones that it can actually infect. Naturally, $m_E(t)$ is always less than $m(t)$, but it still starts at approximately $m_E(0)\approx\mu$ with the population starting out mostly susceptible. We show that $m_E(t)$ also has an initial increase in some parameter regimes, but not in others, followed also by a steady decline as susceptibility decreases.

\begin{table}[!b]
\centering{
\begin{tabular}{ccc}
\hline
\multirow{2}{*}{Variable} & \multirow{2}{*}{Equation} & Peak value \\[-.5mm]
& &  \ (with $\theta(0)\approx1$)\\
\hline
$m(t)$ & $\dot m=-\frac{J}{I}\big(m-\frac{\phi''(\log\theta)}{\phi'(\log\theta)}\big)$ & $\mu+\frac\nu\mu$ \Tstrut\Bstrut\\
$m_E(t)$ & $\dot m_E=-\big(\frac{J}{I}+\beta-\frac{\psi''(\theta)\dot\theta}{\psi'(\theta)}\big)m_E+\frac JI\frac{\psi''(\theta)\theta}{\psi'(1)}$ & $\max\{\mu+\frac\nu\mu-2,\mu\}$ \Tstrut\Bstrut\\
$m_Z(t)$ & $m_Z(t)=\beta\frac{\psi''(\theta(t))}{\psi'(\theta(t))}\int_t^\infty e^{-(\beta+\gamma)(\tau-t)}\frac{\psi'(\theta(\tau))}{\psi'(1)}d\tau$ & $<\RO$ (if $\psi'$ log-convex)\Tstrut\Bstrut\\
\hline
\end{tabular}
}
\resizebox{\linewidth}{!}{
\begin{tabular}{cccc}
\hline
\multirow{2}{*}{Variable} & Peak time & \multicolumn{2}{c}{Value in the limit $t\to\infty$}\\[-.5mm]
& \ (with $\theta(0)\approx1$) & \ if $\psi''(\theta(\infty))>\mu$ & if $\psi''(\theta(\infty))\le\mu$\\
\hline
$m(t)$ & $t_m<\big(\frac12-\frac\gamma{4(\beta+\gamma)(\RO-1)+2\gamma}\big)t_I$ & $\ \frac{\phi''(\log\theta(\infty))}{\phi'(\log\theta(\infty))}$ & $>\frac{\phi''(\log\theta(\infty))}{\phi'(\log\theta(\infty))}$ \Tstrut\Bstrut\\
$m_E(t)$ & $0$ if $\frac\nu\mu>2$, else $>0$ and $<t_m$ & $\ \theta(\infty)\big(\frac{\psi''(\theta(\infty))}{\mu}-1\big)<\frac\gamma\beta$ & 0 \Tstrut\Bstrut\\
$m_Z(t)$ & 0 (if $\psi'$ log-convex) & \multicolumn{2}{c}{$=\frac\beta{\beta+\gamma}\frac{\psi''(\theta(\infty))}{\mu}<1$} \Tstrut\Bstrut\\
\hline
\end{tabular}
}
\caption{Summary of the analytic results derived in this paper concerning three metrics of superspreading at time $t$ we define: the mean $m(t)$ of the infected degree distribution, the mean $m_E(t)$ of the effective degree distribution, and the mean $m_Z(t)$ of the secondary case distribution. Here, $\beta$ is the transmission rate; $\gamma$ is the recovery rate; $I$ is the infection prevalence; $J$ is the instantaneous incidence; $\theta$ is the probability an edge has not yet transmitted infection; $\mathcal{R}_0$ is the basic reproduction number; $t_I$ is the peak time of $I$; and $\psi$ and $\phi$ are the probability- and moment-generating functions, respectively, of the network's degree distribution $K$.}
\label{table:summary}
\end{table}

Third and lastly, we define the \textit{secondary case distribution} $Z(t)$ as the distribution of the number of secondary cases each newly infected node at time $t$ will cause during its infectious period (named following Ref. \cite{LloydSmithEtAl2005}), which has mean $m_Z(t)$, variance $v_Z(t)$, moments $m_{Z,n}(t)$, and mass function $p_{Z,k}(t)$. While this is perhaps the most direct way of characterizing superspreading potential over time, it is also the least analytically tractable. However, we are still able to show that $m_Z(t)$ decreases monotonically at all times under certain assumptions, with no initial increase as with $m(t)$ or $m_E(t)$. This difference occurs because $Z(t)$ is concerned only with \textit{newly} infected nodes at time $t$, whose degrees follow the neighbor degree distribution $K_{\textnormal{n}}$ for small $t$. Conversely, $X(t)$ and $E(t)$ are concerned with all nodes currently infected at $t$, which for small $t$ will still include some of the initial infections whose degrees follow the overall degree distribution $K$.

With these distributions rigorously defined and equations derived for them, we then investigate their means (and higher moments when possible), peak times, and limiting behavior. The main analytic results are summarized in Table~\ref{table:summary}. Notably, we introduce the concept of the \textit{superspreading peak}, defined as the peak of $m(t)$, at the superspreading peak time $t_m$ when $m(t)$ is maximized and the potential for superspreading could be considered to be at its highest. Interestingly, we show that the superspreading peak time $t_m$ is less than half the peak time $t_I$ of prevalence $I(t)$, while $m_E(t)$ peaks before $m(t)$ (and $m_Z(t)$ peaks even earlier at $t=0$). Thus, the potential for superspreading declines much earlier than the epidemic's peak, meaning that much of the impact of superspreading may have already occurred by the time an outbreak has even reached epidemic proportions and become a serious threat.

\subsection*{Overview of the edge-based network epidemic model}

Our analyses are based on Miller and Volz's edge-based model \cite{Miller2011,MillerEtAl2012} for its simplicity, which we now briefly summarize. Miller and Volz describe a SIR epidemic in continuous time on a configuration network \cite{Newman2018} as 
\begin{align}
\dot\theta&=-\beta\theta+\beta\frac{\psi'(\theta)}{\psi'(1)}+\gamma(1-\theta) \label{eq:dtheta}\\
S&=\psi(\theta) \label{eq:S}\\
I&=1-S-R \label{eq:I}\\
\dot R&=\gamma I, \label{eq:dR}
\end{align}
with $S$, $I$, and $R$ representing the fraction of nodes susceptible, infected, and recovered, respectively ($S(t)+I(t)+R(t)=1$), $\beta$ is the transmission rate, $\gamma$ the recovery rate, and
\begin{equation}
\psi(x)=\langle x^K\rangle=\sum\limits_{k=0}^\infty P(k)x^k
\label{eq:psi}
\end{equation}
is the probability-generating function for the network's degree distribution $K$ with probability mass function $P(k)$. The variable $0<\theta(t)<1$ represents the probability that an arbitrary neighbor of a focal susceptible node has not yet passed infection to that node by time $t$ (either because the neighbor was never infected or the neighbor was infected but never transmitted along their connecting edge), which starts out at $\theta(0)\approx 1$ as we assume initial prevalence is small. Miller shows that an outbreak will occur if and only if \begin{equation}
\RO = \frac{\beta}{\beta+\gamma}\frac{\psi''(1)}{\psi'(1)}>1,
\label{eq:R0}
\end{equation}
and then $\theta(t)$ will always be monotonically decreasing towards a final value $\theta(\infty)$. Since $\dot\theta(\infty)=0$, the quantity $\theta(\infty)$ must satisfy
\begin{equation}
\label{eq:theta_infty}
\theta(\infty)=\frac{\gamma}{\beta+\gamma}+\frac{\beta}{\beta+\gamma}\frac{\psi'(\theta(\infty))}{\psi'(1)}.
\end{equation}
From now on, $\theta(\infty)$ will refer to the unique solution to Eq.~\eqref{eq:theta_infty} between 0 and 1 exclusive, which Miller shows must exist. The function $\psi$ also has many useful properties: in this model, susceptibility is simply $S=\psi(\theta)$, while in general, $\psi'(1)=\langle K\rangle=\mu$ and $\psi''(1)=\langle K(K-1)\rangle=\mu_2-\mu=\nu+\mu^2-\mu$, where $\mu_n$ the $n$-th moment of $K$, $\mu=\mu_1$ is its mean, and $\nu$ its variance (these notations are summarized in Table~\ref{table:notation}). It will also be helpful to define the instantaneous incidence of new infectious
\begin{equation}
J=-\dot S=-\psi'(\theta)\dot\theta,
\label{eq:J}
\end{equation}
and the moment-generating function of $K$
\begin{equation}
\phi(y)=\langle e^{yK}\rangle=\sum\limits_{k=0}^\infty P(k)e^{yk}
\label{eq:phi}
\end{equation}
satisfying $\phi(\log x)=\psi(x)$ and $\phi^{(n)}(0)=\mu_n$.

\subsection*{Assumptions}

Our results apply to any $\RO>1$, while we ignore the case of $\RO<1$ as this precludes the possibility of an outbreak occurring. We assume that the initial fraction of infected nodes is very small, which corresponds to an initial value of $\theta(0)$ very close to 1, as many of our results depend on this small initial infection limit. We also assume that at least the first three moments of $K$ exist (excluding some power-law degree distributions from this analysis) so that $\psi'$, $\psi''$, $\psi'''$ and $\phi'$, $\phi''$, $\phi'''$ are all well defined. We also assume that the network is static, large, and created by the configuration model, so that the network has negligible self-loops, multi-edges, degree correlations, clustering, and modularity \cite{Newman2018}.

\subsection*{Example distributions and simulations}

As examples, we investigate epidemics on networks with three different degree distributions: (1) $K$ is Poisson distributed, with no dispersion; (2) $K$ is negative-binomially distributed with low dispersion (dispersion parameter $r=2.5$); and (3) $K$ is negative-binomially distributed with high dispersion ($r=.5$). In Fig.~\ref{fig:Is}, we show what the epidemic trajectories look like for each of these three network degree distributions. We employ negative binomial distributions (and Poisson distributions in the limiting case of no dispersion) as they have been found to be good fits to empirical superspreading event \cite{LloydSmithEtAl2005} and allow exact edge-based description \cite{KissEtAl2023}. For all three, we fix the mean degree $\mu=5$ and the disease's basic reproduction number $\RO=3$. From there, $\beta$ can be calculated from Eq.~\eqref{eq:R0}, while $\gamma$ can be thought of as a time-scale constant and set arbitrarily to $1$ (we do so and thus refer to $t$ as ``$\gamma$-scaled time'').

We construct networks for each of these three degree distributions and with 1 million nodes each using the configuration model and removing self-loops and multi-edges. On each of these three networks, we use code originally written for the ``Epidemics on Networks'' Python package \cite{MillerTing2019} to run 200 simulations, each starting with a random set of 100 infections out of the 1 million nodes.

\begin{figure}[!t]
\centering
\includegraphics[width=\textwidth]{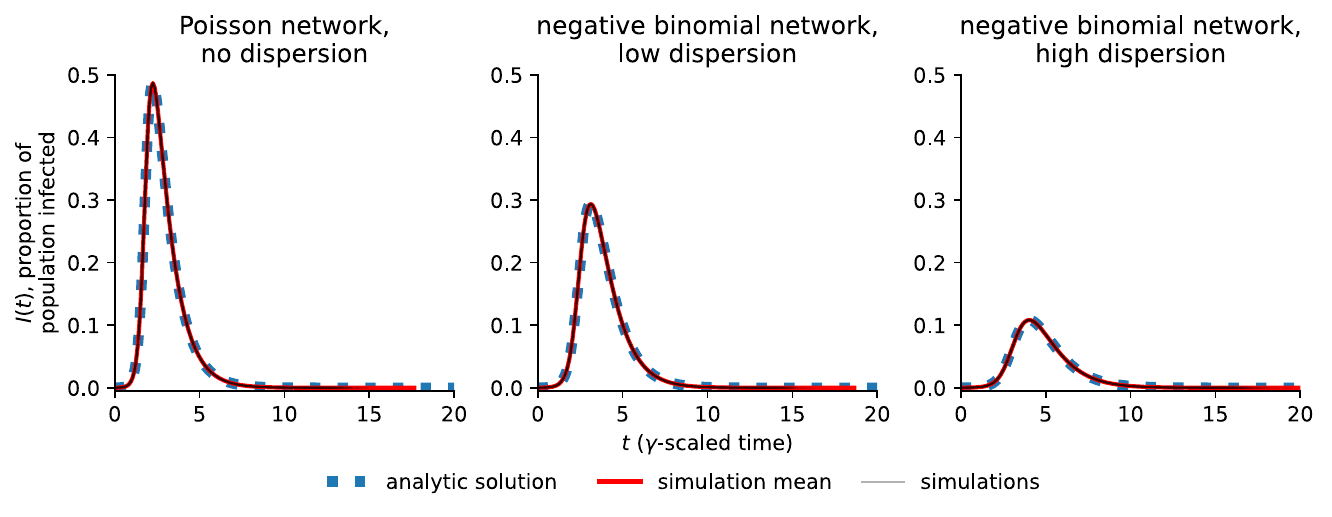}
\caption{Trajectories of infection prevalence in the three networks we use as examples, each with 1 million nodes and mean degree 5: a Poisson network with no dispersion, a negative-binomial network with low dispersion (dispersion parameter $r=2.5$), and a negative-binomial network with high dispersion ($r=.5$). Red curves show the mean of 200 simulations (whose individual trajectories are also plotted by faint gray curves which coincide with the mean curve) while the blue dotted curves show the analytic solution provided by the edge-based model Eq.~\eqref{eq:dtheta}--Eq.~\eqref{eq:dR} \cite{Miller2011}. We scale time $t$ by the infection recovery rate $\gamma$.}
\label{fig:Is}
\end{figure}

\section*{The infected degree distribution \texorpdfstring{$X(t)$}{X(t)}}

To derive the infected degree distribution, it will be helpful to have equations for how overall prevalence $I(t)$ changes over time as well as $I_k(t)$, which we define to be the prevalence of infection among nodes of degree $k$. The first of these follows immediately from differentiating Eq.~\eqref{eq:I} to get
\begin{equation}
\dot I=-\psi'(\theta)\dot\theta-\gamma I.
\label{eq:dI}
\end{equation}
For degree $k$ nodes, and for each of their $k$ neighbors, $\theta$ is the probability that the neighbor has not yet passed infection to the node, so that a degree $k$ node will be susceptible with probability $\theta^k$. From this logic, we get an analog to the original model but for degree $k$ nodes only:
\begin{align}
S_k&=\theta^k \label{eq:Sk}\\
I_k&=1-S_k-R_k \label{eq:Ik}\\
\dot R_k&=\gamma I_k \label{eq:dRk}
\end{align}
and
\begin{equation}
\dot I_k=-k\theta^{k-1}\dot\theta-\gamma I_k,
\label{eq:dIk}
\end{equation}
with the dynamics of $\theta$ still described by Eq.~\eqref{eq:dtheta}.

By Bayes' law, the probability that an infected node is of degree $k$ equals
\begin{equation}
p_k=\frac{P(k)I_k}{I},
\label{eq:ik}
\end{equation}
which we can differentiate to get $\dot p_k=P(k)\frac{I\dot I_k-I_k\dot I}{I^2}$, simplifying to
\begin{equation}
\dot p_k=-\frac{J}{I}\left(p_k-\frac{P(k)k\theta^{k-1}}{\psi'(\theta)}\right),
\label{eq:dpk}
\end{equation}
where again $J=-\psi'(\theta)\dot\theta$ is the instantaneous incidence of new infections. The $p_k$ represent the mass function for the infected degree distribution $X(t)$, which starts out equal to the network degree distribution so that $p_k(0)=P(k)$.

Alternately, we can express $\dot p_k$ in terms of the degree distribution's moment-generating function $\phi$, which instead yields
\begin{equation}
\dot p_k=-\frac{J}{I}\left(p_k-\frac{P(k)k\theta^k}{\phi'(\log\theta)}\right).
\label{eq:dpk2}
\end{equation}
From this formulation, succinct differential equations for the infected degree distribution's moments easily follow.
\begin{theorem}
\label{thm:dmn}
Assume $\mu_{n+1}$, the $(n+1)$-th moment of the network degree distribution, exists. Then the $n$-th moment of the infected degree distribution $X(t)$ exists and satisfies
\begin{equation}
\dot m_n=-\frac{J}{I}\left(m_n-\frac{\phi^{(n+1)}(\log\theta)}{\phi'(\log\theta)}\right),
\label{eq:dmn}
\end{equation}
with initial value of $m_n$ given by
\begin{equation}
m_n(0)=\mu_n.
\label{eq:mn_init}
\end{equation}
\end{theorem}
\begin{proof}
Note that, since $\mu_{n+1}$ exists, $\phi^{(j)}(\log\theta)$ will also exist for any $j\le n+1$ and $0<\theta<1$, as $\phi^{(j)}(\log\theta)=\langle K^j\theta^K\rangle\le\langle K^{n+1}\rangle=\mu_{n+1}$. Eq.~\eqref{eq:dmn} then follows from Eq.~\eqref{eq:dpk2} as
\begin{align*}
\dot m_n &= \sum_{k=0}^\infty \dot p_k k^n\\
&= -\frac{J}{I}\left(\sum_{k=0}^\infty p_kk^n-\sum_{k=0}^\infty\frac{P(k)k^{n+1}\theta^k}{\phi'(\log\theta)}\right)\\
&= -\frac{J}{I}\left(m_n-\frac{\phi^{(n+1)}(\log\theta)}{\phi'(\log\theta)}\right).
\end{align*}

And since the initial infected degree distribution and the network degree distribution are equal, so are their moments, such that $m_n(0)=\mu_n$.
\end{proof}

\begin{figure}[!t]
\centering
\includegraphics[width=\textwidth]{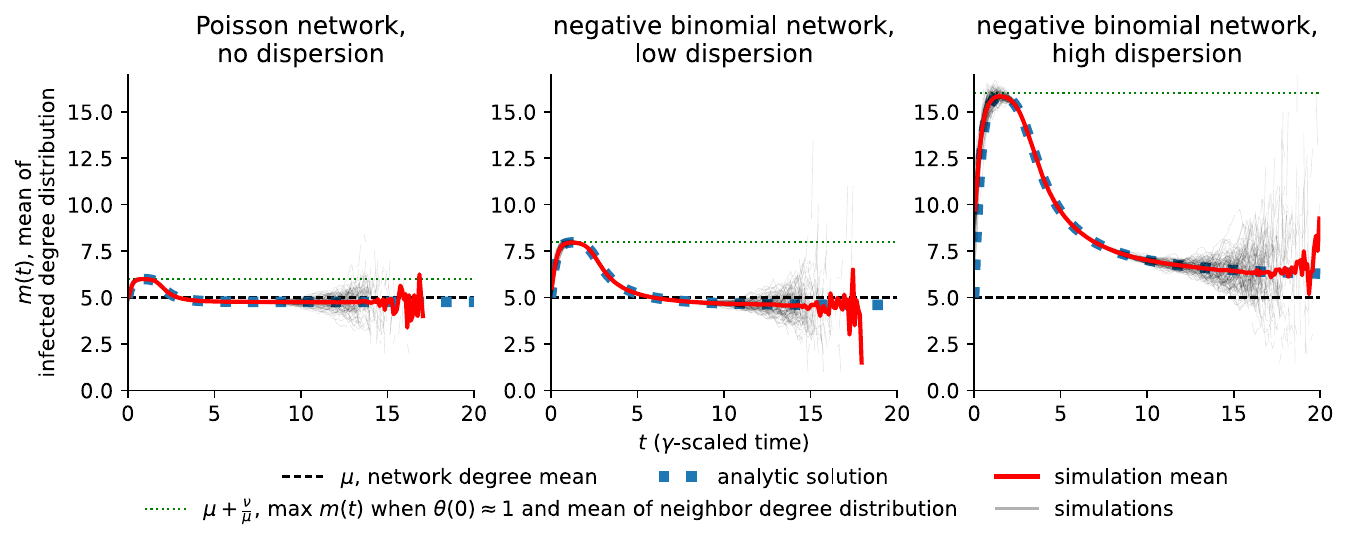}
\caption{Trajectories of the infected degree distribution's mean $m(t)$ for three different network degree distributions. Red curves show the mean of 200 simulations (whose individual trajectories are also plotted by faint gray curves) while the blue dotted curves show the analytic solution from Eq.~\eqref{eq:dm}. Black dashed lines show the mean $\mu$ of the network degree distribution and green dotted lines show the mean $\mu+\frac\nu\mu$ of the neighbor degree distribution.}
\label{fig:ms}
\end{figure}

Intuitively, Eq.~\eqref{eq:dmn} says that $m_n(t)$ is continuously being attracted to the moving target $\frac{\phi^{(n+1)}(\log\theta(t))}{\phi'(\log\theta(t))}$ at a rate proportional to $\frac{J(t)}{I(t)}$, which we call the \textit{infection turnover rate}. This is the relative rate at which current infections are replaced with new ones, thus also controlling the rate at which the degree distribution of infected nodes changes. Before continuing, it will be helpful to establish the behavior of this infection turnover rate, with a lemma we prove in the supplement.
\begin{lemma}
\label{lemma:J_I}
The following results are true for the infection turnover rate $\frac{J(t)}{I(t)}$ when $\theta(0)$ is sufficiently close to 1:
\begin{enumerate}[(1)]
\item $\frac{J(t)}{I(t)}>\gamma$ while $\dot J(t)\ge0$.

\item $\frac{J(t)}{I(t)}<\beta\left(\frac{\psi''(1)}{\psi'(1)}-1\right)$ always, and $\frac{J(t)}{I(t)}$ can be made to stay arbitrarily close to $\beta\left(\frac{\psi''(1)}{\psi'(1)}-1\right)$ for an arbitrary amount of time while $\theta(t)\approx1$ by choosing $\theta(0)$ sufficiently close to 1.

\item $\lim\limits_{t\to\infty}\frac{J(t)}{I(t)}=\max\left\{\beta\left(\frac{\psi''(\theta(\infty))}{\psi'(1)}-1\right),0\right\}$.
\end{enumerate}
\end{lemma}
We illustrate this behavior of infection turnover rate $\frac{J(t)}{I(t)}$ in supplementary Fig.~S4.

Now, we are specifically interested in the mean and variance of the infected degree distribution, which by Theorem~\ref{thm:dmn} for $n=1$ satisfies
\begin{equation}
\dot m=\frac{J}{I}\left(m-\frac{\phi''(\log\theta)}{\phi'(\log\theta)}\right)
\label{eq:dm}
\end{equation}
with $m(0)=\mu$, and $v(t)=\mu_2(t)-\mu(t)^2$ with $\mu_2(t)$ also defined by Theorem~\ref{thm:dmn}. We show the trajectories of $m(t)$ in Fig.~\ref{fig:ms} and the trajectories of $v(t)$ in supplementary Fig.~S1.

The mean infected degree $m(t)$ is seen to first increase from the mean of the network degree distribution $\mu$, peak at approximately the mean of the neighbor degree distribution $\mu+\frac\nu\mu$, then decline monotonically. This is in fact true of all moments of $X(t)$: $m_n(t)$ starts out at the $n$-th moment of the network degree distribution $\mu_n$, peaks at approximately the $n$-th moment of the neighbor degree distribution $\frac{\phi^{(n+1)}(0)}{\phi(0)}$, then declines monotonically. We call this peak of $m(t)$, when the infected degree distribution is approximately equal to the neighbor degree distribution, the \textit{superspreading peak}. This notion can be formalized by the following result.

\begin{theorem}
In the limit as $\theta(0)\approx 1$, the infected degree distribution $X(t)$ will approach the neighbor degree distribution $K_{\textnormal{n}}$ at some time $t$ (which may change with $\theta(0)$ as it approaches 1). Specifically, for any small $\epsilon>0$ and $k\in\N$, there exist a $\theta(0)$ and time $t_\epsilon$ for which $\abs{p_j(t_\epsilon)-\frac{P(k)k}{\mu}}<\epsilon$ for all $j\le k$.
\label{thm:X_KN}
\end{theorem}
\begin{proof}
We provide an outline of the proof here and complete the details in the supplement. By Eq.~\eqref{eq:dpk2}, $p_k(t)$ is constantly moving toward the moving target $\frac{P(k)k\theta(t)^k}{\phi'(\log\theta(t))}$ at rate proportional to the infection turnover rate $\frac{J(t)}{I(t)}$. Conversely, $\theta(t)$ decreases monotonically and is bounded below by
\begin{equation}
\label{eq:theta_ineq}
\theta(t)>1-(1-\theta(0))e^{(\beta+\gamma)(\RO-1)t}\ \text{ for all }t>0,
\end{equation}
which follows from Eq.~\eqref{eq:dtheta} via
\begin{equation}
\label{eq:dlog_1_minus_theta}
\frac d{dt}\log(1-\theta(t))=(\beta+\gamma)(\RO-1)-\frac{\beta}{\psi'(1)}\left(\psi''(1)-\frac{\psi'(1)-\psi'(\theta(t))}{1-\theta(t)}\right)
\end{equation}
and from the convexity of $\psi'$ (since $\psi'''$ is assumed to exist and is positive). Thus, for any fixed value $\theta_J>\theta(\infty)$, by setting $\theta(0)$ sufficiently close to 1 we can ensure that there is a time $t_J$ for which $\theta(t_J)=\theta_J$ and this time can be made arbitrarily large by setting $\theta(0)$ sufficiently close to 1. Specifically, we choose $t_J>0$ to be the time at which incidence $J(t)$ peaks, which we show must occur and that $\theta_J=\theta(t_J)$ is constant regardless of initial condition. Furthermore, Lemma~\ref{lemma:J_I} tells us that $\frac{J(t)}{I(t)}>\gamma$ for all $t\le t_J$. With this constant positive lower bound on the infection turnover rate $\frac{J(t)}{I(t)}$ for $t\le t_J$ and with the ability to make $t_J$ arbitrarily large by choosing $\theta(0)$ sufficiently close to 1, we can then choose a $\theta(0)$ that will keep $\theta(1)$ close to 1 and $\frac{J(t)}{I(t)}>\gamma$ long enough for $p_j(t_J)$ to reach sufficiently close to $\frac{P(k)k}\mu$ for all $j\le k$.
\end{proof}

The peak values of $m(t)$ and higher moments $m_n(t)$ come as a direct corollary.
\begin{corollary}
In the limit as $\theta(0)\to 1$, $m_n(t)$ peaks at value
\begin{equation}
\lim\limits_{\theta(0)\to1}\max\limits_{t\ge0}m_n(t)=\frac{\phi^{(n+1)}(0)}{\phi'(0)}.
\label{eq:mn_peak}
\end{equation}
In particular,
\begin{equation}
\lim\limits_{\theta(0)\to1}\max\limits_{t\ge0}m(t)=\mu+\frac\nu\mu.
\label{eq:m_peak}
\end{equation}
\end{corollary}

We are also interested in the time $t_m$ it takes for the mean infected degree $m(t)$ to peak, called the \textit{superspreading peak time}, and its relation to the peak time $t_I$ of prevalence $I(t)$. Interestingly, we find that the peak time for $m(t)$ is always less than half the peak time of $I(t)$ when $\theta(0)$ is sufficiently close to 1, suggesting that the potential for superspreading is already diminished by the time infections in an epidemic have taken off. In the next result, we provide an upper bound on the ratio $\frac{t_m}{t_I}$ in the limit of $\theta(0)\to1$ that is even tighter than $\frac12$, while noting that in this limit $\frac{t_m}{t_I}$ stays greater than 0 despite $t_I$ going to $\infty$.

\begin{theorem}
\label{thm:tm}
If $\theta(0)$ is sufficiently close to 1, then $m(t)$ and $I(t)$ will both peak at times $t_m>0$ and $t_I>0$ respectively. And as $\theta(0)$ approaches 1 both $t_m$ and $t_I$ will diverge to $\infty$ while
\begin{equation}
0<\lim_{\theta(0)\to1}\frac{t_m}{t_I}<\frac12-\frac{\gamma}{4(\beta+\gamma)(\RO-1)+2\gamma}.
\label{eq:tm}
\end{equation}
\end{theorem}
\begin{proof}
We provide an outline of the proof here and complete the details in the supplement. We have already shown in Theorem~\ref{thm:X_KN} that $t_m$ exists and that $t_J$ goes to infinity as $\theta(0)$ approaches 1, and we show in the supplement that $t_m$ also goes to infinity. And the peak $t_I>0$ of $I(t)$ also exists from the fact that $I(0)$ is close to zero and $\RO>1$ by assumption \cite{Miller2011}. Furthermore, $\lim\limits_{\theta(0)\to1}\frac{t_I}{t_J}=1$, as we show in the supplement. Thus, we can replace $t_I$ with $t_J$ in Eq.~\eqref{eq:tm}, which greatly simplifies the analysis since, as we have previously mentioned, $\theta(t_J)=\theta_J$ is always constant regardless of initial conditions.

We now consider $1-\theta(t)$ and bound it above and below by exponential functions for $t\le t_J$. Above, it is bounded by an exponential with rate $\lambda_1=(\beta+\gamma)(\RO-1)$ from Eq.~\eqref{eq:theta_ineq}. And below it is bounded by an exponential with rate $\lambda_2=\lambda_1-\frac\beta{\psi'(1)}\left(\psi''(1)-\frac{\psi'(1)-\psi'(\theta_J)}{1-\theta_J}\right)$ from Eq.~\eqref{eq:dlog_1_minus_theta} and since $\frac{\psi'(1)-\psi'(x)}{1-x}$ is increasing for $x<1$ from the convexity of $\psi'$. Since $\lambda_1>\lambda_2>0$, as we prove, then
\begin{equation}
0<\frac1{\lambda_1}<\frac{t_J}{\log\frac{1}{1-\theta(0)}+\log(1-\theta_J)}<\frac1{\lambda_2}.
\label{eq:tJ_ineq}
\end{equation}

To derive $t_m$, note that the $\psi''(1)-\frac{\psi'(1)-\psi'(\theta(t))}{1-\theta(t)}$ term in Eq.~\eqref{eq:dlog_1_minus_theta} vanishes uniformly across all $t\le t_m$ as $\theta(0)\to 1$ (since we have shown that $\lim\limits_{\theta(0)\to\infty}\theta(t_m)=1$), especially as compared to $\log(1-\theta(t))$ whose magnitude becomes arbitrary large for $t\le t_m$ as $\theta(0)\to 1$. Thus, the inequality in Eq.~\eqref{eq:theta_ineq} approaches equality uniformly across all $t\le t_m$ in the limit as $\theta(0)\to1$, and
\begin{equation}
\lim\limits_{\theta(0)\to1}\dot m(t)=-\frac{J(t)}{I(t)}\left(m-\big(f(1)-f'(1)(1-\theta(0))e^{(\beta+\gamma)(\RO-1)t}\big)\right)\text{ for all }t\le t_m
\label{eq:dm_approx}
\end{equation}
converges uniformly across $t\le t_m$, with $f(x)=\frac{\phi''(\log x)}{\phi'(\log x)}$. From this limiting behavior for $m(t)$ and by Lemma~\ref{lemma:J_I}, we can show
\begin{equation}
\lim\limits_{\theta(0)\to1}\frac{t_m}{\log\frac1{1-\theta(0)}}=\frac1{2\lambda_1+\gamma}.
\label{eq:tm_limit}
\end{equation}
Finally, as $\theta(0)\to1$ and $\log\frac1{1-\theta(0)}\to\infty$, we see that $t_m$ diverges, and dividing Eq.~\eqref{eq:tm_limit} by Eq.~\eqref{eq:tJ_ineq} then replacing $t_J$ with $t_I$ in the limit as $\theta(0)\to1$ gives the desired result.
\end{proof}

In Fig.~\ref{fig:peak_ts}, we show how the superspreading peak time $t_m$ compares to this upper bound and to the prevalence peak time $t_I$.

Lastly, we aim to find the limiting behavior of the infected degree distribution as $t\to\infty$. Eq.~\eqref{eq:dmn} shows that $m(t)$ will eventually move towards the value $\frac{\phi^{(n+1)}(\theta(\infty))}{\phi'(\theta(\infty))}$, but this will only be reached in the limit of $t\to\infty$ when the infection turnover rate $\frac{J(t)}{I(t)}$ stays greater than 0, which we know occurs if and only if $\psi''(\theta(\infty))>\mu$ by Lemma~\ref{lemma:J_I}. Similarly, Eq.~\eqref{eq:dpk2} shows that $p_k(t)$ will reach $\frac{P(k)k\theta(\infty)^k}{\phi'(\log\theta(\infty))}$ as $t\to\infty$ if and only if $\psi''(\theta(\infty))>\mu$.  This immediately gives the result
\begin{theorem}
If $\psi''(\theta(\infty))>\mu$, then
\begin{equation*}
\lim\limits_{t\to\infty}m_n(t)=\frac{\phi^{(n+1)}(\theta(\infty))}{\phi'(\theta(\infty))}
\end{equation*}
for any $n$ for which $\mu_{n+1}$, the $(n+1)$-th moment of the network degree distribution, exists; and
\begin{equation*}
\lim\limits_{t\to\infty}p_k(t)=\frac{P(k)k\theta(\infty)^k}{\phi'(\log\theta(\infty))}
\end{equation*}
for all $k$. 

Otherwise, if $\psi''(\theta(\infty))\le\mu$, then
\begin{equation*}
\lim\limits_{t\to\infty}m_n(t)>\frac{\phi^{(n+1)}(\theta(\infty))}{\phi'(\theta(\infty))}
\end{equation*}
for any $n$ for which $\mu_{n+1}$ exists.
\label{thm:X_infty}
\end{theorem}

Before moving on to deriving the effective degree distribution and the secondary case distribution, we provide the examples of infected degree distributions at the superspreading peak and at the end of the epidemic for Poisson and negative binomial networks. These follow directly from Theorems~\ref{thm:X_KN} and \ref{thm:X_infty}.

\begin{figure}
\centering
\includegraphics[width=\linewidth]{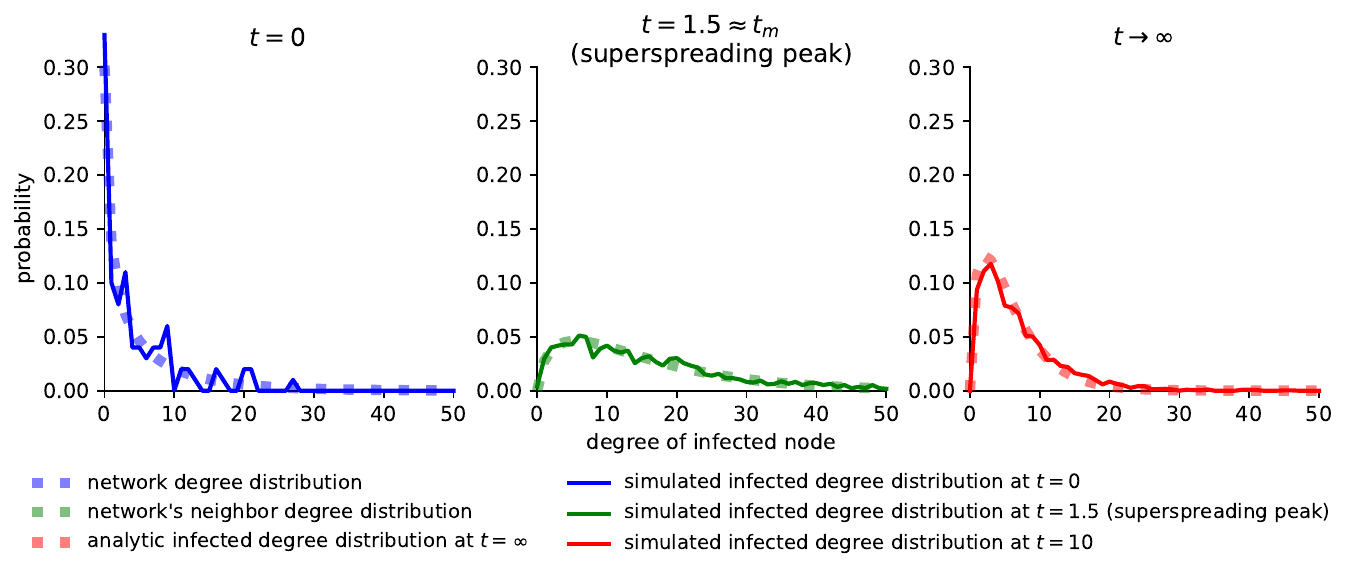}
\caption{The infected degree distribution at the beginning, superspreading peak, and end of an epidemic on a negative binomial distribution with high dispersion (dispersion parameter $r=.5$). Solid lines showing infected degree distributions from a single simulation and dashed lines showing the analytic infected degree distributions given by the network degree distribution $K$ at $t=0$, by Eq.~\eqref{eq:NB_tm} at the superspreading peak, and by Eq.~\eqref{eq:NB_infty} for $t\to\infty$.}
\label{fig:dists}
\end{figure}

\begin{examples*}
If the network degree distribution is Poisson with rate parameter $\lambda$ ($K\sim\Pois(\lambda)$), then at the superspreading peak $t_m$, the infected degree distribution $X$ follows
\begin{equation}
X(t_m)-1\approx\Pois(\lambda).
\end{equation}
If $\psi''(\theta(\infty))>\mu$, then as $t\to\infty$,
\begin{equation}
X(\infty)-1\sim\Pois\big(\theta(\infty)\lambda\big).
\end{equation}

If the network degree distribution is negative binomial with dispersion parameter $r$ and success probability parameter $p$ ($K\sim\NB(r,p)$), then at the superspreading peak $t_m$,
\begin{equation}
X(t_m)-1\approx\NB(r+1,p).
\label{eq:NB_tm}
\end{equation}
If $\psi''(\theta(\infty))>\mu$, then as $t\to\infty$,
\begin{equation}
X(\infty)-1\sim\NB\big(r+1,1-\theta(\infty)(1-p)\big).
\label{eq:NB_infty}
\end{equation}
\end{examples*}
We demonstrate this result in Fig.~\ref{fig:dists} for the negative binomial distribution with high dispersion (dispersion parameter $r=.5$, as this is the only network we examine that satisfies $\psi''(\theta(\infty))>\mu$, as shown by supplementary Fig.~S4).

\section*{The effective degree distribution \texorpdfstring{$E(t)$}{E(t)}}

Since some of an infected node's neighbors may be susceptible at any time, and thus unable to have infection transmitted to them, it is important to consider not only how many neighbors an infected node has, but also how many of those neighbors are susceptible. We call the distribution of the number susceptible neighbors that infected nodes have at time $t$ the \textit{effective degree distribution} $E(t)$. By definition, $E(t)\le X(t)$.

Deriving this distribution is less straightforward, but can be done with the help of new variables $J_j(t)$, the instantaneous incidence of newly infected degree-$j$ nodes at time $t$, and $H_j(t)$, the probability that a degree $j$ node is currently infected at time $t$ and has not transmitted infection to a specific, arbitrary neighbor. Just like with $J(t)$, $J_j(t)$ can easily be derived as
\begin{equation}
J_j=-\dot S_j=-j\theta^{j-1}\dot\theta.
\label{eq:Jj}
\end{equation}
And just as $\dot I_j=J_j-\gamma I_j$ for $I(t)$ since degree $j$ nodes become infected at rate $J_j$ and stop being infected at rate $\gamma$, we now have
\begin{equation}
\dot H_j=J_j-(\beta+\gamma)H_j,\ H_j(0)=I(0),
\label{eq:Hj}
\end{equation}
since a node stops having the property ``infected but has not yet transmitted infection to an arbitrary neighbor'' at rate $\beta+\gamma$ (by either infecting the neighbor or recovering). Initially, $H_j(0)=I_j(0)=I(0)$ since there has been no time yet for any transmissions to happen yet at $t=0$.

To figure out how many susceptible neighbors an infected node of degree $j$ has at time $t$, we must consider two cases: (1) the node was an initial infected node or (2) it was not an initial infected node. The probability that a degree-$j$ infected node is of the first case is $\frac{P(j)e^{-\gamma t}I(0)}{I(t)}$ by Bayes' law, since $e^{-\gamma t}I(0)$ is the probability that an initial infected node is still infected by time $t$. The probability that a degree-$j$ node is of the second case is thus $p_j-\frac{P(j)e^{-\gamma t}I(0)}{I(t)}=\frac{P(j)(I_j(t)-e^{-\gamma t}I(0))}{I(t)}$.

For the first case, consider an initial infected node of degree $j$ which is still infected at time $t$. The probability that one of its neighbors with degree $l$ has not yet had infection transmitted to it by any of its other $l-1$ neighbors is $\theta(t)^{l-1}$. Thus, the probability that a degree-$l$ neighbor is still susceptible is equal to $\theta(t)^{l-1}$ times the probability that the initial infected node has not transmitted infection to this neighbor, which is $e^{-\beta t}$. Then considering that a neighbor will be degree $l$ with probability $\frac{P(l)l}{\psi'(1)}$ according to the neighbor degree distribution, the probability that an arbitrary neighbor of the initial infected node is susceptible at time $t$ is $\sum_{l=0}^\infty\frac{P(l)l\theta(t)^{l-1}}{\psi'(1)}e^{-\beta t}=\frac{\psi'(\theta(t))}{\psi'(1)}e^{-\beta t}$, and the probability that the initial infected node has exactly $k$ susceptible neighbors is
\begin{equation}
\eta_{1,j,k}(t)=\binom jk\left(\frac{\psi'(\theta(t))}{\psi'(1)}e^{-\beta t}\right)^k\left(1-\frac{\psi'(\theta(t))}{\psi'(1)}e^{-\beta t}\right)^{j-k}.
\label{eq:eta1}
\end{equation}

For the second case, consider an infected node of degree $j$ at time $t$ which was not infected initially. The calculation is similar to the first case, except now the infected node can only have a maximum of $j-1$ susceptible neighbors, instead of $j$ as before, since one of the infected node's neighbors must have infected it. And now the probability that this non-initial infected node has not yet transmitted infection to an arbitrary neighbor, given its non-initial infected state, is $\frac{H_j(t)-e^{-(\beta+\gamma)t}I(0)}{I_j(t)-e^{-\gamma t}I(0)}$, with $H_j(t)-e^{-(\beta+\gamma)tI(0)}$ representing the probability that a degree $j$ node is infected but has not yet transmitted to an arbitrary neighbor and was not an initial infected. From this, we get the probability that this non-initial infected node has exactly $k$ susceptible neighbors to be
\begin{equation}
\begin{aligned}
\eta_{2,j,k}(t)=&\binom{j-1}k\left(\frac{\psi'(\theta(t))}{\psi'(1)}\frac{H_j(t)-e^{-(\beta+\gamma)t}I(0)}{I_j(t)-e^{-\gamma t}I(0)}\right)^k\times\\
&\quad\left(1-\frac{\psi'(\theta(t))}{\psi'(1)}\frac{H_j(t)-e^{-(\beta+\gamma)t}I(0)}{I_j(t)-e^{-\gamma t}I(0)}\right)^{j-1-k}.
\end{aligned}
\label{eq:eta2}
\end{equation}

\begin{figure}[!t]
\centering
\includegraphics[width=\textwidth]{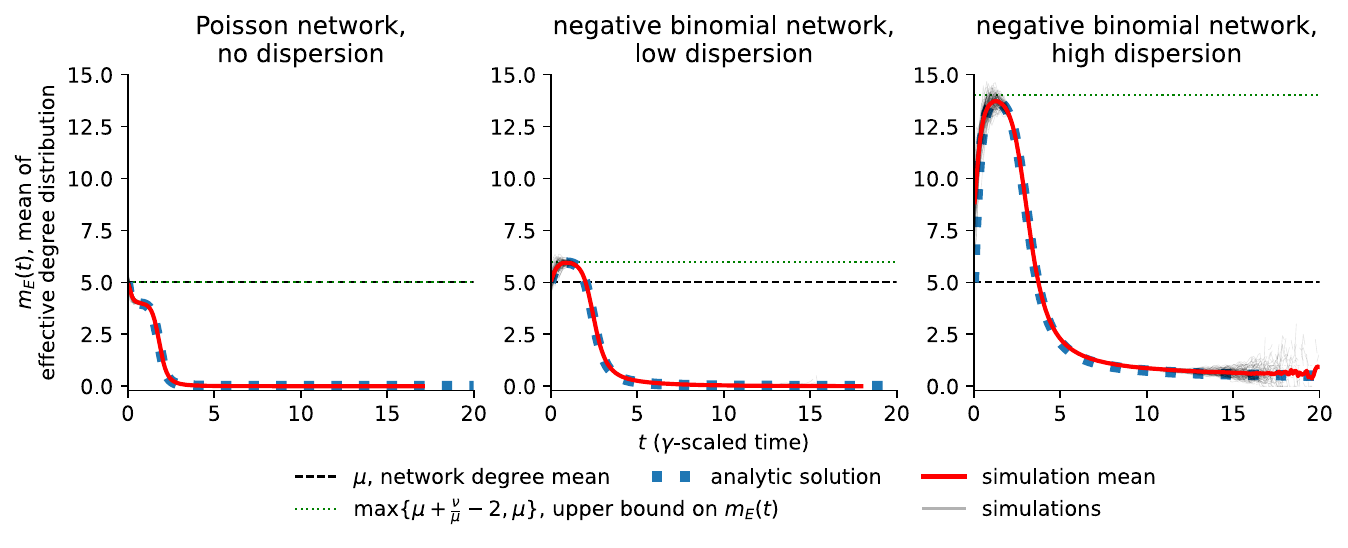}
\caption{Trajectories of the effective degree distribution's mean $m_E(t)$ for three different network degree distributions. Red curves show the mean of 200 simulations (whose individual trajectories are also plotted by faint gray curves) while the blue dotted curves show the analytic solution from Eq.~\eqref{eq:dmE}. Black dashed lines show the mean $\mu$ of the network degree distribution and green dotted lines show the upper bound $\max\{\mu+\frac\nu\mu-2,\mu\}$ provided by Theorem~\ref{thm:mE_peak}.}
\label{fig:mEs}
\end{figure}

Putting this all together, we arrive at the mass function $p_{E,k}(t)$ for the effective degree distribution $E(t)$:
\begin{equation}
p_{E,k}(t)=\sum_{j=k}^\infty\frac{P(j)e^{-\gamma t}I(0)}{I(t)}\eta_{1,j,k}(t)+\sum_{j=k+1}^\infty\frac{P(j)(I_j(t)-e^{-\gamma t}I(0))}{I(t)}\eta_{2,j,k}(t).
\label{eq:pEk}
\end{equation}
While this expression is not particularly nice, the mean $m_E(t)=\sum_{k=0}^\infty p_{E,k}(t)$ of $E(t)$ simplifies to
\begin{equation}
m_E(t)=\frac{\psi'(\theta(t))}{\psi'(1)}\left[\sum_{j=1}^\infty P(j)(j-1)\frac{H_j(t)}{I(t)} + \frac{e^{-(\beta+\gamma)t}I(0)}{I(t)}(1-P(0))\right]
\label{eq:mE}
\end{equation}
where $P(0)$ is the probability a node in the network has degree 0, as the sums of the binomial terms simplify to the expectations of binomial random variables. However, we can do even better by considering how $m_E(t)$ changes over time, expressing it with the following differential equation:
\begin{theorem}
The mean $m_E(t)$ of the effective degree distribution $E(t)$ satisfies
\begin{equation}
\dot m_E=-\left(\frac {J}I+\beta-\frac{\psi''(\theta)\dot\theta}{\psi'(\theta)}\right)m_E+\frac {-\psi'(\theta)\dot\theta}I\frac{\psi''(\theta)\theta}{\psi'(1)},
\label{eq:dmE}
\end{equation}
with initial value given by
\begin{equation}
m_E(0)=\psi'(\theta(0))\approx\mu.
\label{eq:mE0}
\end{equation}
\label{thm:dmE}
\end{theorem}
\begin{proof}
Eq.~\eqref{eq:dmE} follows from differentiating Eq.~\eqref{eq:mE}. And Eq.~\eqref{eq:mE0} follows since $H_j(0)=I(0)$ for all $j$, so $m_E(0)=\frac{\psi'(\theta(t))}{\psi'(1)}\left[\sum_{j=1}^\infty P(j)(j-1) +(1-P(0))\right]=\psi'(\theta(0))$, which is approximately $\psi'(1)=\mu$ for $\theta(0)\approx1$.
\end{proof}

The higher moments of $E(t)$ do not yield such nice differential equations, but we do provide a formula for the second moment $m_{E,2}(t)$ in the supplement, from which the variance can be calculated as $v_E(t)=m_{E,2}(t)-m_E(t)^2$. We show the trajectories of $m_E(t)$ in Fig.~\ref{fig:mEs} and the trajectories of $v_E(t)$ in supplementary Fig.~S2.

The next result focuses on the peak behavior of $m_E(t)$, which we show occurs at $t=0$ in networks where $\frac\nu\mu\le2$ but occurs later when $\frac\nu\mu>\mu$. In general, $\max\{\mu+\frac\nu\mu-2,\mu\}$ is an upper bound to $m_E(t)$, which $m_E(t)$ achieves in the limit as $\theta(0)\to1$, and the peak time of $m_E(t)$ always occurs before the superspreading peak time $t_m$ of $m(t)$.
\begin{theorem}
If $\frac\nu\mu\le2$, then in the limit as $\theta(0)\to1$, $m_E(t)$ will peak at time $t=0$ and at value $m_E(0)=\mu$. If $\frac\nu\mu>2$, then in the limit as $\theta(0)\to1$, $m_E(t)$ will peak at a time greater than 0 but less than the superspreading peak time $t_m$ of $m(t)$ and at the value $\mu+\frac\nu\mu-2$. Thus,
\begin{equation}
\lim\limits_{\theta(0)\to1}\max\limits_{t\ge 0} m_E(t)=\max\left\{\mu+\frac\nu\mu-2,\mu\right\}.
\end{equation}

The peak time of $m_E(t)$ is less than $t_m$ if $\theta(0)$ is sufficiently close to 1, and will occur at $t=0$ if $\frac\nu\mu\le2$. And the value of $m_E(t)$ at its peak is bounded above by
\begin{equation}
 \max\limits_{t\ge0}m_E(t)<\max\{\mu+\frac\nu\mu-2,\mu\}.
\label{eq:mE_peak}
\end{equation}
\label{thm:mE_peak}
\end{theorem}
\begin{proof}
Reformatting Eq.~\eqref{eq:dmE} as
\begin{equation}
\dot m_E=-\left(\frac JI+\beta-\frac{\psi''(\theta)}{\psi'(\theta)}\dot\theta\right)\left(m_E-\frac{J/I}{J/I+\beta-\frac{\psi''(\theta)}{\psi'(\theta)}\dot\theta}\frac{\psi''(\theta)\theta}{\psi'(1)}\right)
\label{eq:dmE2}
\end{equation}
shows that $m_E$ is constantly attracted towards the moving target $\frac{J/I}{J/I+\beta-\frac{\psi''(\theta)}{\psi'(\theta)}\dot\theta}\frac{\psi''(\theta)\theta}{\psi'(1)}$ at rate $\frac JI+\beta-\frac{\psi''(\theta)}{\psi'(\theta)}\dot\theta$. Since $\dot\theta<0$ and $\frac JI<\beta\left(\frac{\psi''(1)}{\psi'(1)}-1\right)$ by Lemma~\ref{lemma:J_I}, this moving target will always be bounded above by
\begin{equation}
\frac{J/I}{J/I+\beta-\frac{\psi''(\theta)}{\psi'(\theta)}\dot\theta}\frac{\psi''(\theta)\theta}{\psi'(1)}<\left(1-\frac{\beta}{J/I+\beta}\right)\frac{\psi''(1)}{\psi'(1)}<\frac{\psi''(1)}{\psi'(1)}-1=\mu+\frac\nu\mu-2.
\label{eq:mE_target}
\end{equation}
If $\frac\nu\mu\le2$, then the moving target will always to $\mu=\lim\limits_{\theta(0)\to1}m_E(0)$, and so in the limit as $\theta(0)\to1$, $m_E(t)$ will peak at $t=0$ at value $m_E(0)=\mu$.

Now assume that $\frac\nu\mu>2$. Lemma~\ref{lemma:J_I} gives us that $\frac JI$ approaches $\beta\left(\frac{\psi''(1)}{\psi'(1)}-1\right)$ from below at the beginning and can be made to stay arbitrarily close to this value for arbitrarily long while $\theta(t)\approx1$ by choosing $\theta(0)$ sufficiently close 1. Thus, in the limit as $\theta(0)\to 1$, when calculating the target $m_E(t)$ moves towards, we can consider $\theta\approx1$ and $\dot\theta\approx0$ as constants, and $\frac JI$ as being less than but approaching $\beta\left(\frac{\psi''(1)}{\psi'(1)}-1\right)$, so that Eq.~\eqref{eq:mE_target} will always be true but in the limit as $\theta(0)\to 1$ the inequalities in Eq.~\eqref{eq:mE_target} will tend toward equalities.
\end{proof}

This result is reflected in Fig.~\ref{fig:mEs}, which shows that $m_E(t)$ always stays under this upper bound of $\max\{\mu+\frac\nu\mu-2,\mu\}$, and in Fig.~\ref{fig:peak_ts}, which shows that the peak times of $m_E(t)$ are always less than the superspreading peak times $t_m$ (and consequently also bounded above by half the prevalence peak time $t_I$). Furthermore, in the case of the Poisson network where $\frac\nu\mu=1\le 2$, we see that $m_E(t)$ does indeed peak at $t=0$, while in the other two cases were $\frac\nu\mu>2$, we see that $m_E(t)$ peaks later and at a value approximately equal to $\mu+\frac\nu\mu-2$. This peak value of $\mu+\frac\nu\mu-2$ for $m_E(t)$ (when $\frac\nu\mu>2$) is significant, since this is 1 less then the value $\mu+\frac\nu\mu-1$ which is the mean ``excess degree'' of the network (for a node reach by following a random edge, the excess degree is the number of other edges that node has \cite{Newman2018}). In other words, non-initial infected nodes always have on average one less susceptible neighbor than they could at their maximum.

We now examine the behavior of $m_E(t)$ as $t\to\infty$, seeing that it settles at a positive value if and only if $\psi''(\theta(\infty))<\mu$ (which we note is also the same condition that governs the final behavior of $m(t)$), otherwise $m_E(t)$ vanishes to 0.
\begin{theorem}
\begin{equation}
\lim\limits_{t\to\infty}m_E(t)=\max\left\{\theta(\infty)\left(\frac{\psi''(\theta(\infty))}\mu-1\right),0\right\}<\frac\gamma\beta.
\label{eq:mE_infty}
\end{equation}
\label{thm:mE_infty}
\end{theorem}
\begin{proof}
Since the rate $\frac JI+\beta-\frac{\psi''(\theta)}{\psi'(\theta)}\dot\theta$ is always at least $\beta>0$, $m_E(t)$ will always reach the final value of its moving target $\frac{J/I}{J/I+\beta-\frac{\psi''(\theta)}{\psi'(\theta)}\dot\theta}\frac{\psi''(\theta)\theta}{\psi'(1)}$ in the limit as $t\to\infty$. Noting that $\lim\limits_{t\to\infty}\frac{J(t)}{I(t)}=\max\left\{\beta\left(\frac{\psi''(\theta(\infty))}{\psi'(1)}-1\right),0\right\}$ (Lemma~\ref{lemma:J_I}) and $\lim\limits_{t\to\infty}\dot\theta(t)=0$, and then simplifying the moving target accordingly, gives the desired result.

To see that $\theta(\infty)\left(\frac{\psi''(\theta(\infty))}\mu-1\right)<\frac\gamma\beta$, realize that $\psi''(\theta(\infty))<\frac{\psi'(1)-\psi'(\theta(\infty))}{1-\theta(\infty)}$ by the convexity of $\psi'$ so that
\begin{equation}
\frac{\beta}{\beta+\gamma}\frac{\psi''(\theta(\infty))}{\psi'(1)}<\frac{\beta}{\beta+\gamma}\frac{\psi'(1)-\psi'(\theta(\infty))}{\psi'(1)(1-\theta(\infty))}=1,
\label{eq:infty_less_than_1}
\end{equation}
the last step of which follows from Eq.~\eqref{eq:theta_infty}. The desired inequality than immediately follows.
\end{proof}
This result is demonstrated in Fig.~\ref{fig:mEs}, where only the negative binomial network with high dispersion satisfies $\psi''(\theta(\infty))>\mu$ (as shown by supplementary Fig.~S4), and this is also the only network for which $m_E(t)$ is seen to stay above 0 as $t\to\infty$.

\section*{The secondary case distribution \texorpdfstring{$Z(t)$}{Z(t)}}

Lastly, we introduce the \textit{secondary case distribution} $Z(t)$, which we define to be the distribution of secondary cases a node infected at time $t$ will produce over the course of its infectious period. $Z(t)$ most directly captures superspreading: what is really important is not that some people are connected more than others, or that they expose susceptible individuals more than others, but rather that they \textit{infect} more than others. This is often how superspreading is defined in the literature in more simple branching process models, as the distribution of secondary cases each infection causes \cite{LloydSmithEtAl2005}. However, such models typically ignore time-varying susceptibility and finite population sizes, thus making our formulation of the secondary case distribution on network epidemics novel. Although $Z(t)$ is in some sense the most natural definition of superspreading we consider, it is also the least analytically tractable, and we are unable to derive a simple differential equation formulation for even its mean $m_Z(t)$.

We also note that $Z(t)$ is defined based on newly infected nodes at time $t$ (thus excluding initially infected nodes), while the other distributions $X(t)$ and $E(t)$ are defined for all nodes currently infected at time $t$ (though in the supplement we also define and explore analog distributions to $X(t)$ and $E(t)$ which are defined for newly infected nodes at time $t$ instead). This difference is necessary for defining the secondary case distribution since it is concerned with all future transmissions an infected node causes. It thus makes sense to start tracking these transmissions as soon as the node becomes infected; the infected degree distribution and the effective degree distribution, however, are only concerned with the current state of an infected node and its neighbors at time $t$, which is important to track at different points of an infected node's infectious period. As we will show, this distinction has important consequences for the behavior of $Z(t)$.

\begin{figure}[!t]
\centering
\includegraphics[width=\textwidth]{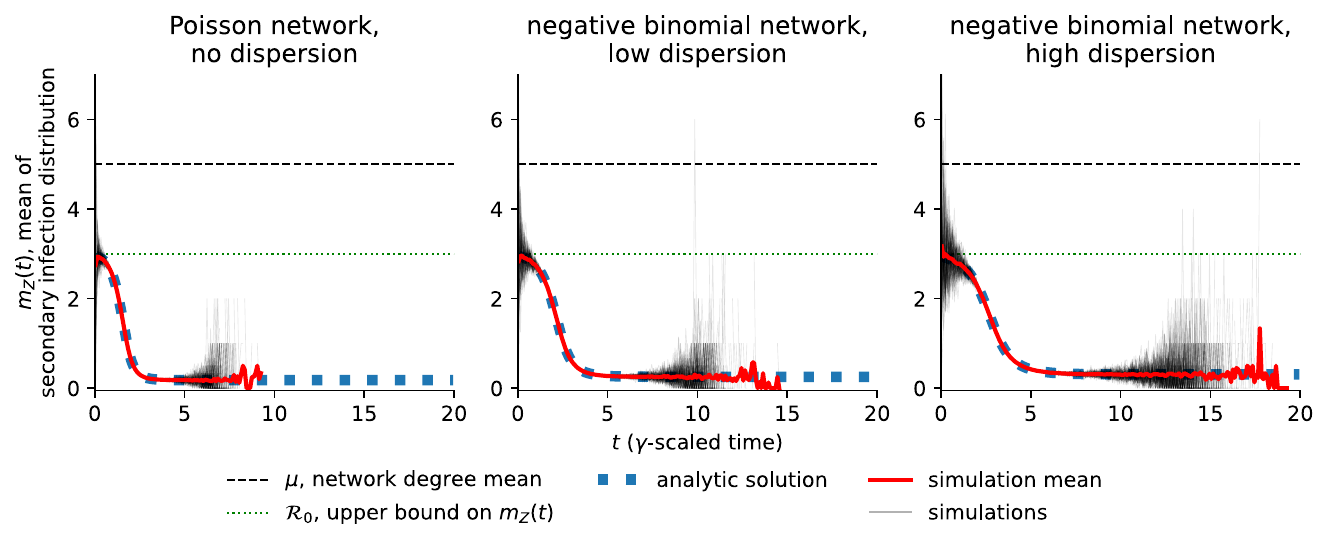}
\caption{Trajectories of the secondary case distribution's mean $m_Z(t)$ for three different network degree distributions. Red curves show the mean of 200 simulations (whose individual trajectories are also plotted by faint gray curves) while the blue dotted curves show the analytic solution from Eq.~\eqref{eq:mZ}. Black dashed lines show the mean $\mu$ of the network degree distribution and green dotted lines show the basic reproduction number $\RO$.}
\label{fig:mZs}
\end{figure}

In order to derive its mass function $p_z(t)$, we first define the quantity $\zeta_t^k(s)$, which is the probability that a node $x$ which becomes infected at time $t>0$ and recovers at time $s$ infects an arbitrary degree-$k$ neighbor $y$ (other than the neighbor which infected $x$) by time $s$. This can be calculated using Bayes' law as follows, noting that $y$ cannot have infected by $x$ by time $t$ and $x$ cannot have infected $y$ by time $t$:
 \begin{align*}
&\zeta_t^k(s)=\P\big(y\text{ susceptible at time }t\mid y\text{ never infected }x\text{ by time }t\big)\ \times\\
&\quad\int_t^s\P\big(x\text{ infects }y\text{ in time }(\tau,\tau+d\tau)\mid y\text{ not infected by any other neighbor by time }\tau\big)\ \times\\
&\hspace{9.5mm}\P\big(y\text{ not infected by another other neighbor by time }\tau\mid y\text{ susceptible at time t}\big)\\
&\phantom{\zeta_t^k(s)}=\frac{\theta(t)^{k-1}}{\theta(t)}\int_t^s\beta e^{-\beta(\tau-t)}\frac{\theta(\tau)^{k-1}}{\theta(t)^{k-1}}\ d\tau\\
&\phantom{\zeta_t^k(s)}=\frac1{\theta(t)}\int_t^s\beta e^{-\beta(\tau-t)}\theta(\tau)^{k-1}\ d\tau.
\end{align*}
From this, we then define $\zeta_t(s)$ to be the same quantity but with $y$ an arbitrary neighbor of $x$ of any degree, so that its degree $k$ is now distributed according to the neighbor degree distribution $\frac{P(k)k}{\mu}$. Then
\begin{equation}
\zeta_t(s)=\sum_{k=0}^\infty\frac{P(k)k}{\mu}\zeta_t^k(s)=\frac1{\theta(t)}\int_t^s\beta e^{-\beta(\tau-t)}\frac{\psi'(\theta(\tau))}{\psi'(1)}\ d\tau.
\label{eq:zeta}
\end{equation}

Now consider a node $x$ newly infected at time $t$ and with degree $j$. One of its $j$ neighbors must be the one that transmitted infection to it, so $x$ can transmit infection to a maximum of $j-1$ of its neighbors. And for each of these $j-1$ neighbors there is probability $\zeta_t(s)$ that $x$ will infect each neighbor, given that $x$ recovers at time $s$. Since the probability of $x$ recovering at a time in the range $(s,s+ds)$ is $\gamma e^{-\gamma(s-t)}ds$, then the probability that $x$ infects exactly $k$ of its neighbors during its infectious period is $\int_t^\infty\gamma e^{-\gamma(s-t)}\binom{j-1}{k}\zeta_t(s)^k(1-\zeta_t(s))^{j-1-k}ds$.

Now this $x$ newly infected at time $t$ will have degree $j$ with probability $\frac{J_j(t)}{J(t)}=\frac{j\theta(t)^{j-1}}{\psi'(\theta(t))}$, since $J(t)$ is the number of new infections at time $t$ while $J_j$ is the number of those which have degree $j$. Finally, this gives the probability a node newly infected at time $t$ will infect exactly $k$ neighbors during its infectious period as
\begin{equation}
p_{Z,k}(t)=\sum_{j=k+1}^\infty\frac{j\theta(t)^{j-1}}{\psi'(\theta(t))}\int_t^\infty\gamma e^{-\gamma(s-t)}\binom{j-1}{k}\zeta_t(s)^k(1-\zeta_t(s))^{j-1-k}ds.
\label{eq:pZk}
\end{equation}

With the mass function derived for the secondary case distribution now derived, we now show that its mean $m_Z(t)$ can be expressed more nicely, but still falls short of a simple differential equation formulation due to the presence of integrals from $t$ to infinity.
\begin{theorem}
The mean $m_Z(t)$ of the secondary case distribution $Z(t)$ satisfies
\begin{equation}
m_Z(t)=\frac{\psi''(\theta(t))}{\psi'(\theta(t))}\int_t^\infty\beta e^{-(\beta+\gamma)(\tau-t)}\frac{\psi'(\theta(\tau))}{\psi'(1)}\ d\tau,
\label{eq:mZ}
\end{equation}
and if the function $\psi'$ is log-convex, then the initial value of $m_Z$ satisfies
\begin{equation}
m_Z(0)<\RO.
\end{equation}
\label{thm:mZ}
\end{theorem}
\begin{proof}
The equation for the mean follows by plugging Eq.~\eqref{eq:pZk} into $m_Z(t)=\sum_{k=0}^\infty p_{E,k}(t)$, with the sums of the binomial terms simplifying to the expectations of binomial random variables. Now the condition that $\psi'$ is log-convex implies $\frac{\psi''(\theta(t))}{\psi'(\theta(t))}<\frac{\psi''(1)}{\psi'(1)}$ for all $t$, so that $m_Z(0)<\frac{\psi''(1)}{\psi'(1)}\int_0^\infty\beta e^{-(\beta+\gamma)\tau}\ d\tau=\frac{\beta}{\beta+\gamma}\frac{\psi''(1)}{\psi'(1)}=\RO$.
\end{proof}

We note that this log-convexity condition for $\psi'$ will only ever be met by network degree distributions with infinite support, but most distributions commonly used to model contact networks that have defined variance (precluding scale-free networks) have this log-convexity property for the derivative of their probability-generating function, including Poisson, negative binomial, and geometric distributions.

The higher moments of $Z(t)$ are even less tractable, but we provide a formula for the second moment $m_{Z,2}(t)$ in the supplement, from which the variance can be calculated as $v_Z(t)=m_{Z,2}(t)-m_Z(t)^2$. We show the trajectories of $m_Z(t)$ in Fig.~\ref{fig:mZs} and the trajectories of $v_Z(t)$ in supplementary Fig.~S3.

The peak behavior of $m_Z(t)$ is not as interesting as with $m(t)$ or $m_E(t)$. Again, assuming log-convexity of $\psi'$, it is easy to see from Eq.~\eqref{eq:mZ} that $m_Z(t)$ will always be decreasing:
\begin{theorem}
Assuming $\psi'$ is log-convex, $m_Z(t)$ will always be decreasing for all $t$.
\end{theorem}
\begin{proof}
This follows from the facts that $\frac{\psi''(\theta(t))}{\psi'(\theta(t))}$ is decreasing, by the log-convexity of $\psi'$, and that $\psi'(\theta(\tau))$ is decreasing with $\tau$.
\end{proof}
This result makes intuitive sense if $m_Z(t)$ is compared to the notion of the ``effective reproduction number'' $\mathcal{R}(t)$, which is typically defined to be the average number of secondary cases produced by individuals infected at time $t$ over the course of their infectious period if population susceptibility were to stay constant from time $t$ onward. This last condition of susceptibility staying constant is the only difference conceptually between $m_Z(t)$ and $\mathcal{R}(t)$, and this should ensure that $m_Z(t)<\mathcal{R}(t)$ for all $t$. Thus, it makes sense that $m_Z(t)$ declines for all $t$ and $m_Z(0)<\RO$, since $\mathcal{R}(t)$ generally declines for all $t$ in epidemics without demographics or waning immunity and $\mathcal{R}(0)\approx\RO$ by definition.

As mentioned before, $Z(t)$ is fundamentally different from $X(t)$ and $E(t)$ in its focus on newly infected nodes at time $t$ rather than at all currently infected nodes at time $t$. This distinction is responsible for the difference in peak behavior between $m_Z(t)$ and the means $m(t)$ and $m_E(t)$ of the other distributions. At the start, most of the nodes $X(t)$ and $E(t)$ consider are initial infections, whose average degree ($\mu$) is less than the average degree of the newly infected nodes near the start ($\mu+\frac\nu\mu$). However, $Z(t)$ only considers these newly infected nodes, and its calculation never considers the initial infections.

Finally, we provide the behavior of $m_Z(t)$ as $t\to \infty$, which always converges to a positive final value that is less than 1:
\begin{theorem}
\begin{equation}
\lim_{t\to\infty}m_Z(t)=\frac{\beta}{\beta+\gamma}\frac{\psi''(\theta(\infty))}{\mu}<1.
\label{eq:mZ_infty}
\end{equation}
\label{thm:mZ_infty}
\end{theorem}
\begin{proof}
In the limit as $t\to\infty$, we substitute $\theta(\infty)$ in for $\theta(t)$ in Eq.~\eqref{eq:mZ}, which gives $\lim\limits_{t\to\infty}m_Z(t)=\lim\limits_{t\to\infty}\frac{\psi''(\theta(\infty))}{\psi'(\theta(\infty))}\frac{\psi'(\theta(\infty))}{\psi'(1)}\int_t^\infty\beta e^{-(\beta+\gamma)(\tau-t)}d\tau$ and thus the desired result. And we have previously shown that this quantity is less than 1 in Eq.~\eqref{eq:infty_less_than_1}.
\end{proof}
This fact also aligns with the closely related effective reproduction number $\mathcal{R}$, which also always reaches a final value less than 1 at the end of a non-endemic epidemic.

\begin{figure}[!t]
\centering
\includegraphics[width=\textwidth]{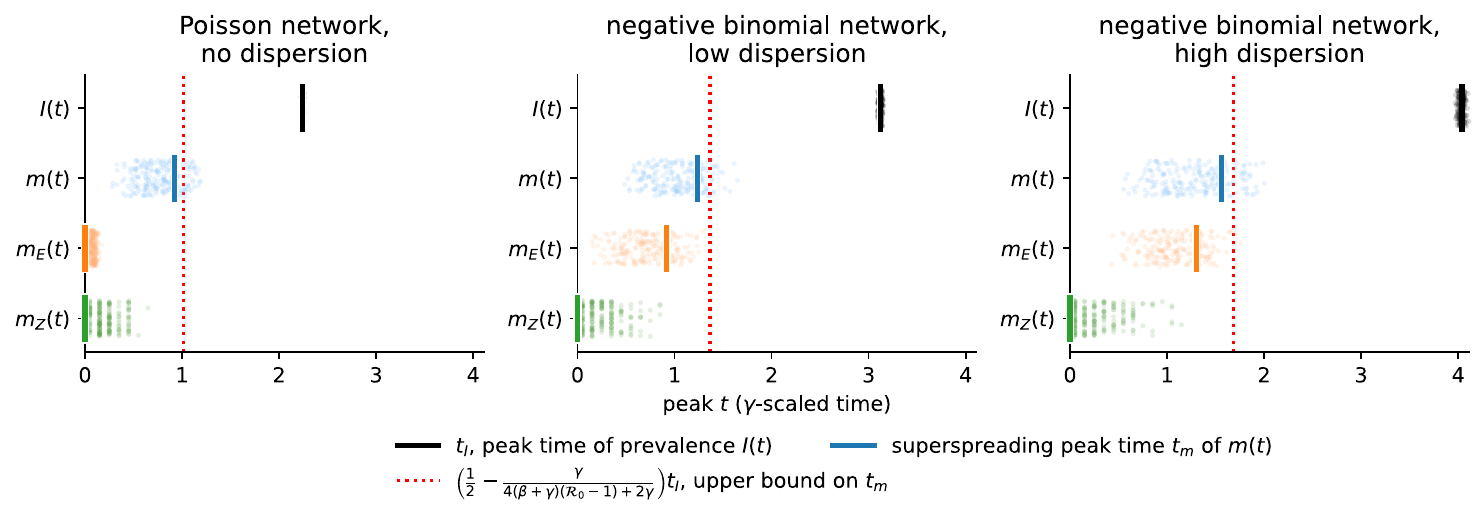}
\caption{Peak times of $I(t)$, $m(t)$, $m_E(t)$, and $m_Z(t)$ for three different network degree distributions, and from both simulation (faint points in background, binned into multiples of .1 for $m_Z(t)$) and analytic (solid lines) trajectories. The dashed red line shows the upper bound on the superspreading peak time $t_m$ of $m(t)$ given by Theorem~\ref{thm:tm}, demonstrating that all three superspreading metrics we consider peak sooner than half the time it takes for infection prevalence to peak. When measuring simulation peak times, we ignore the highly stochastic period at the end of the epidemics where there are very few infections and these metrics may have wild fluctuations (as visible in Figs.~\ref{fig:ms}, \ref{fig:mEs}, \ref{fig:mZs}).}
\label{fig:peak_ts}
\end{figure}

\begin{figure}[!t]
\centering
\includegraphics[width=\textwidth]{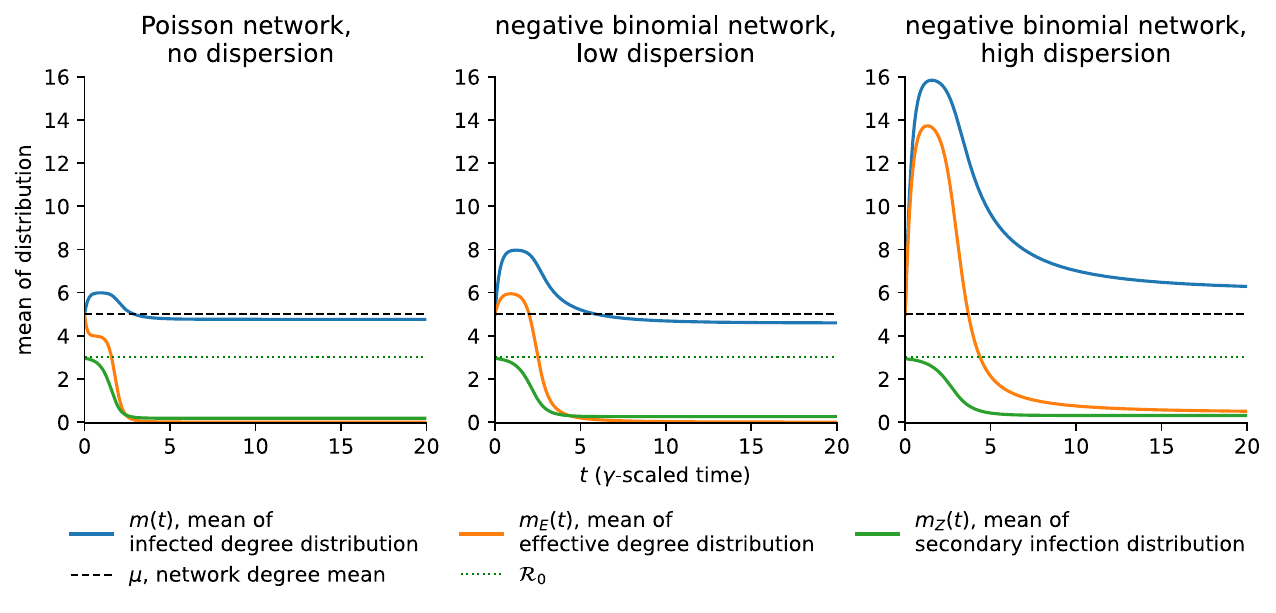}
\caption{Analytic trajectories of the three superspreading metrics---the mean $m(t)$ of the infected degree distribution, the mean $m_E(t)$ of the effective degree distribution, and the mean $m_Z(t)$ of the secondary case distribution---we define in Eq.~\eqref{eq:dm}, Eq.~\eqref{eq:dmE}, and Eq.~\eqref{eq:mZ}, respectively, and for three different network degree distributions. Black dashed lines show the mean $\mu$ of the network degree distribution and green dotted lines show the basic reproduction number $\RO$.}
\label{fig:means}
\end{figure}

\section*{Discussion}

While it has been known in both the theoretical and empirical literature that the potential for superspreading in an epidemic declines over time, we provide the first rigorous mathematical demonstration of this fact in a model of a network epidemic, specifically utilizing the edge-based framework of Miller and Volz \cite{Miller2011,MillerEtAl2012}. Depending on how the ``potential for superspreading'' is defined, however, can yield different qualitative results, as we show. We define three reasonable metrics for tracking superspreading over time: the mean $m(t)$ of the infected degree distribution, the mean $m_E(t)$ of the effective degree distribution, and the mean $m_Z(t)$ of the secondary case distribution (shown all together in Fig.~\ref{fig:means}). These metrics measure the expected number of contacts involving infectious individuals, expected exposures involving susceptible and infectious individuals, and expected transmissions; in order of increasing accordance with the definition of superspreading but also in order of decreasing analytic tractability. Yet, we are still able to show that these three metrics differ in their behaviors in crucial ways. Notably, while all three metrics do decline eventually, signaling a decreasing role of superspreading, $m(t)$ and $m_E(t)$ can have an initial increase and noticeable peaks. We show that the peak times of these metrics all occur in less than half the time it takes for population-level prevalence to peak in an epidemic, suggesting that the role of superspreading declines well before an epidemic reaches its most severe.

This result implies that the efficacy of contact-based control strategies \cite{NielsenEtAl2021,SneppenEtAl2021,KainEtAl2021, BoudreauEtAl2023} to minimize the role of superspreading are highly dependent on the timing of the intervention. Specifically, trying to control superspreading events via contact-based interventions may have little effect once the potential for superspreading has died down. And, as we have shown, this happens rather early compared to the overall dynamics of the epidemic, for all metrics we consider (Fig.~\ref{fig:peak_ts}). After infections have reached a considerable level in the population, it may be more effective to switch to other strategies that aim to uniformly curb transmission among all individuals rather than trying to target those with more contacts.

Our results also have important implications for the accuracy of methods that aim to measure and quantify superspreading. Most commonly, dispersion is measured by estimating the distribution of secondary cases caused by infected individuals: This can be done directly with transmission trees inferred from contact tracing or sequencing data \cite{AdamEtAl2020}, or by simulating epidemic models to reproduce observed incidence \cite{RiouAlthaus2020}, which can give similar results \cite{HebertDufresneEtAl2021}. These methods correspond to estimating the secondary case distribution, $Z(t)$. On the other hand, some methods instead attempt to measure dispersion by estimating the distribution of the total number of contacts infected individuals have (some of which may not lead to transmissions), through the use of mobility data \cite{LauEtAl2020}. This corresponds more to estimating the infected degree distribution $X(t)$. However, as we have shown, the temporal properties and magnitudes of $Z(t)$ and $X(t)$ can differ substantially. Thus, parameter inference may depend greatly on the distribution underlying the method, whether it is concerned with infected individuals' secondary transmissions or their total contacts.

Furthermore, we show that these distributions can change drastically over the course of an epidemic, especially in the early stages, as population susceptibility shifts from higher degree to lower degree nodes. Thus, regardless of the method used to infer dispersion, the data available is likely to involve a significant time period over which the importance of superspreading will have changed. Using fine-grained temporal data (on incidence or contact tracing) might limit the statistical power of the method, but coarse-graining the data involves averaging over very different superspreading patterns. Consequently, it is common to separate an epidemic in multiple periods, sometimes simply in half with two phases of rising or decreasing incidence \cite{KoEtAl2022}. In our simple framework, we show that the second half of the timeframe (and therefore more than half of available case data) will show significantly lower superspreading than what drove the early epidemic dynamics. Future work could use our framework to redefine epidemic phases based on the varying importance of superspreading over time. In doing so, we could redesign inference methods, forecasting models, and intervention strategies to better adapt to the time-varying statistical patterns of epidemics.

\section*{Code and data availability}

All code to run the analyses and produce the figures in this paper can be found at \url{https://github.com/freedmanari/infected_degrees}.

\section*{Acknowledgments}

A.S.F., M.M.N., and S.A.L. were supported by the National Science Foundation grant CCF1917819. A.S.F. acknowledges funding support from the National Science Foundation under grant DMS-2436120 to the University of Vermont. Funding for S.A.L. was provided by the National Science Foundation grant DMS-2327711. B.F.N. acknowledges financial support from the Carlsberg Foundation (grants CF23-0173 and CF24-1337). L.H.-D. acknowledges financial support from the National Institutes of Health 1P20 GM125498-01 Centers of Biomedical Research Excellence Award.

\section*{Competing Interests}

The authors have no competing interests to declare.

\clearpage
\singlespacing
\bibliographystyle{style}

\bibliography{references}

@article{Gleeson2013,
	title = {Binary-{State} {Dynamics} on {Complex} {Networks}: {Pair} {Approximation} and {Beyond}},
	volume = {3},
	shorttitle = {Binary-{State} {Dynamics} on {Complex} {Networks}},
	url = {https://link.aps.org/doi/10.1103/PhysRevX.3.021004},
	doi = {10.1103/PhysRevX.3.021004},
	number = {2},
	urldate = {2026-01-14},
	journal = {Physical Review X},
	author = {Gleeson, James P.},
	month = apr,
	year = {2013},
	note = {Publisher: American Physical Society},
	pages = {021004},
}

@article{MarceauEtAl2010,
	title = {Adaptive networks: {Coevolution} of disease and topology},
	volume = {82},
	number = {3},
	journal = {Physical Review E},
	author = {Marceau, Vincent and Noël, Pierre-André and Hébert-Dufresne, Laurent and Allard, Antoine and Dubé, Louis J},
	year = {2010},
	note = {Publisher: APS},
	pages = {036116},
}

@article{HebertDufresneEtAl2021,
	title = {The network epidemiology of an {Ebola} epidemic},
	journal = {arXiv preprint arXiv:2111.08686},
	author = {Hébert-Dufresne, Laurent and Young, Jean-Gabriel and Bedson, Jamie and Skrip, Laura A and Pedi, Danielle and Jalloh, Mohamed F and Raulier, Bastian and Lapointe-Gagné, Olivier and Jambai, Amara and Allard, Antoine and {others}},
	year = {2021},
}

@article{RiouAlthaus2020,
	title = {Pattern of early human-to-human transmission of {Wuhan} 2019 novel coronavirus (2019-{nCoV}), {December} 2019 to {January} 2020},
	volume = {25},
	issn = {1560-7917},
	url = {https://www.eurosurveillance.org/content/10.2807/1560-7917.ES.2020.25.4.2000058},
	doi = {10.2807/1560-7917.ES.2020.25.4.2000058},
	language = {en},
	number = {4},
	urldate = {2026-01-13},
	journal = {Eurosurveillance},
	author = {Riou, Julien and Althaus, Christian L.},
	month = jan,
	year = {2020},
	note = {Publisher: European Centre for Disease Prevention and Control},
	pages = {2000058},
}

@article{HebertDufresneEtAl2010,
	title = {Propagation dynamics on networks featuring complex topologies},
	volume = {82},
	number = {3},
	journal = {Physical Review E},
	author = {Hébert-Dufresne, Laurent and Noël, Pierre-André and Marceau, Vincent and Allard, Antoine and Dubé, Louis J},
	year = {2010},
	note = {Publisher: APS},
	pages = {036115},
}

@article{KoEtAl2022,
	title = {Secondary transmission of {SARS}-{CoV}-2 during the first two waves in {Japan}: {Demographic} characteristics and overdispersion},
	volume = {116},
	journal = {International Journal of Infectious Diseases},
	author = {Ko, Yura K and Furuse, Yuki and Ninomiya, Kota and Otani, Kanako and Akaba, Hiroki and Miyahara, Reiko and Imamura, Tadatsugu and Imamura, Takeaki and Cook, Alex R and Saito, Mayuko and {others}},
	year = {2022},
	note = {Publisher: Elsevier},
	pages = {365--373},
}

@article{AdamEtAl2020,
	title = {Clustering and superspreading potential of {SARS}-{CoV}-2 infections in hong kong},
	volume = {26},
	number = {11},
	journal = {Nature Medicine},
	author = {Adam, Dillon C and Wu, Peng and Wong, Jessica Y and Lau, Eric HY and Tsang, Tim K and Cauchemez, Simon and Leung, Gabriel M and Cowling, Benjamin J},
	year = {2020},
	note = {Publisher: Nature Publishing Group US New York},
	pages = {1714--1719},
}

@article{BoudreauEtAl2023,
	title = {Temporal and probabilistic comparisons of epidemic interventions},
	volume = {85},
	number = {12},
	journal = {Bulletin of Mathematical Biology},
	author = {Boudreau, Mariah C and Allen, Andrea J and Roberts, Nicholas J and Allard, Antoine and Hébert-Dufresne, Laurent},
	year = {2023},
	note = {Publisher: Springer},
	pages = {118},
}

@article{KissEtAl2023,
	title = {Necessary and sufficient conditions for exact closures of epidemic equations on configuration model networks},
	volume = {87},
	issn = {1432-1416},
	url = {https://doi.org/10.1007/s00285-023-01967-9},
	doi = {10.1007/s00285-023-01967-9},
	language = {en},
	number = {2},
	urldate = {2026-01-12},
	journal = {Journal of Mathematical Biology},
	author = {Kiss, István Z. and Kenah, Eben and Rempała, Grzegorz A.},
	month = aug,
	year = {2023},
	pages = {36},
}

@article{AlthouseEtAl2020,
	title = {Superspreading events in the transmission dynamics of {SARS}-{CoV}-2: {Opportunities} for interventions and control},
	volume = {18},
	issn = {1545-7885},
	shorttitle = {Superspreading events in the transmission dynamics of {SARS}-{CoV}-2},
	url = {https://journals.plos.org/plosbiology/article?id=10.1371/journal.pbio.3000897},
	doi = {10.1371/journal.pbio.3000897},
	language = {en},
	number = {11},
	urldate = {2025-12-04},
	journal = {PLOS Biology},
	author = {Althouse, Benjamin M. and Wenger, Edward A. and Miller, Joel C. and Scarpino, Samuel V. and Allard, Antoine and Hébert-Dufresne, Laurent and Hu, Hao},
	month = nov,
	year = {2020},
	note = {Publisher: Public Library of Science},
	pages = {e3000897},
}

@article{KarrerNewman2010,
	title = {Message passing approach for general epidemic models},
	volume = {82},
	doi = {10.1103/PhysRevE.82.016101},
	number = {1},
	journal = {Physical Review E},
	author = {Karrer, Brian and Newman, M. E. J.},
	year = {2010},
}

@article{SherborneEtAl2018,
	title = {Mean-field models for non-{Markovian} epidemics on networks},
	volume = {76},
	issn = {1432-1416},
	url = {https://doi.org/10.1007/s00285-017-1155-0},
	doi = {10.1007/s00285-017-1155-0},
	language = {en},
	number = {3},
	urldate = {2025-11-01},
	journal = {Journal of Mathematical Biology},
	author = {Sherborne, Neil and Miller, Joel C. and Blyuss, Konstantin B. and Kiss, Istvan Z.},
	month = feb,
	year = {2018},
	pages = {755--778},
}

@book{Newman2018,
	address = {Oxford, New York},
	edition = {2nd},
	title = {Networks},
	isbn = {978-0-19-880509-0},
	publisher = {Oxford University Press},
	author = {Newman, Mark},
	month = sep,
	year = {2018},
}

@article{KainEtAl2021,
	title = {Chopping the tail: {How} preventing superspreading can help to maintain {COVID}-19 control},
	volume = {34},
	issn = {1755-4365},
	shorttitle = {Chopping the tail},
	url = {https://www.sciencedirect.com/science/article/pii/S1755436520300487},
	doi = {10.1016/j.epidem.2020.100430},
	urldate = {2025-05-23},
	journal = {Epidemics},
	author = {Kain, Morgan P. and Childs, Marissa L. and Becker, Alexander D. and Mordecai, Erin A.},
	month = mar,
	year = {2021},
	pages = {100430},
}

@article{MillerTing2019,
	title = {{EoN} ({Epidemics} on {Networks}): a fast, flexible {Python} package for simulation, analytic approximation, and analysis of epidemics on networks},
	volume = {4},
	issn = {2475-9066},
	shorttitle = {{EoN} ({Epidemics} on {Networks})},
	url = {https://joss.theoj.org/papers/10.21105/joss.01731},
	doi = {10.21105/joss.01731},
	language = {en},
	number = {44},
	urldate = {2025-05-22},
	journal = {Journal of Open Source Software},
	author = {Miller, Joel C. and Ting, Tony},
	month = dec,
	year = {2019},
	pages = {1731},
}

@article{PastorSatorrasVespignani2002,
	title = {Epidemic dynamics in finite size scale-free networks},
	volume = {65},
	url = {https://link.aps.org/doi/10.1103/PhysRevE.65.035108},
	doi = {10.1103/PhysRevE.65.035108},
	number = {3},
	urldate = {2025-05-21},
	journal = {Physical Review E},
	author = {Pastor-Satorras, Romualdo and Vespignani, Alessandro},
	month = mar,
	year = {2002},
	note = {Publisher: American Physical Society},
	pages = {035108},
}

@article{AllardEtAl2023,
	title = {The {Role} of {Directionality}, {Heterogeneity}, and {Correlations} in {Epidemic} {Risk} and {Spread}},
	volume = {65},
	issn = {0036-1445},
	url = {https://epubs.siam.org/doi/10.1137/20M1383811},
	doi = {10.1137/20M1383811},
	number = {2},
	urldate = {2025-05-21},
	journal = {SIAM Review},
	author = {Allard, Antoine and Moore, Cristopher and Scarpino, Samuel V. and Althouse, Benjamin M. and Hébert-Dufresne, Laurent},
	month = may,
	year = {2023},
	note = {Publisher: Society for Industrial and Applied Mathematics},
	pages = {471--492},
}

@article{Andersson1997,
	title = {Epidemics in a population with social structures},
	volume = {140},
	issn = {0025-5564},
	url = {https://www.sciencedirect.com/science/article/pii/S0025556496001290},
	doi = {10.1016/S0025-5564(96)00129-0},
	number = {2},
	urldate = {2025-05-21},
	journal = {Mathematical Biosciences},
	author = {Andersson, Håkan},
	month = mar,
	year = {1997},
	pages = {79--84},
}

@article{MayLloyd2001,
	title = {Infection dynamics on scale-free networks},
	volume = {64},
	url = {https://link.aps.org/doi/10.1103/PhysRevE.64.066112},
	doi = {10.1103/PhysRevE.64.066112},
	number = {6},
	urldate = {2025-05-21},
	journal = {Physical Review E},
	author = {May, Robert M. and Lloyd, Alun L.},
	month = nov,
	year = {2001},
	note = {Publisher: American Physical Society},
	pages = {066112},
}

@article{KissEtAl2005,
	title = {Infectious disease control using contact tracing in random and scale-free networks},
	volume = {3},
	url = {https://royalsocietypublishing.org/doi/10.1098/rsif.2005.0079},
	doi = {10.1098/rsif.2005.0079},
	number = {6},
	urldate = {2025-05-21},
	journal = {Journal of The Royal Society Interface},
	author = {Kiss, Istvan Z and Green, Darren M and Kao, Rowland R},
	month = aug,
	year = {2005},
	note = {Publisher: Royal Society},
	pages = {55--62},
}

@book{DiekmannEtAl2012,
	title = {Mathematical {Tools} for {Understanding} {Infectious} {Disease} {Dynamics}},
	isbn = {978-1-4008-4562-0},
	language = {en},
	publisher = {Princeton University Press},
	author = {Diekmann, Odo and Heesterbeek, Hans and Britton, Tom},
	month = nov,
	year = {2012},
	note = {Google-Books-ID: XbntAQAAQBAJ},
}

@article{GrossmannEtAl2021,
	title = {Heterogeneity matters: {Contact} structure and individual variation shape epidemic dynamics},
	volume = {16},
	issn = {1932-6203},
	shorttitle = {Heterogeneity matters},
	url = {https://journals.plos.org/plosone/article?id=10.1371/journal.pone.0250050},
	doi = {10.1371/journal.pone.0250050},
	language = {en},
	number = {7},
	urldate = {2025-05-21},
	journal = {PLOS ONE},
	author = {Großmann, Gerrit and Backenköhler, Michael and Wolf, Verena},
	month = jul,
	year = {2021},
	note = {Publisher: Public Library of Science},
	pages = {e0250050},
}

@article{LeventhalEtAl2015,
	title = {Evolution and emergence of infectious diseases in theoretical and real-world networks},
	volume = {6},
	copyright = {2015 The Author(s)},
	issn = {2041-1723},
	url = {https://www.nature.com/articles/ncomms7101},
	doi = {10.1038/ncomms7101},
	language = {en},
	number = {1},
	urldate = {2022-08-19},
	journal = {Nature Communications},
	author = {Leventhal, Gabriel E. and Hill, Alison L. and Nowak, Martin A. and Bonhoeffer, Sebastian},
	month = jan,
	year = {2015},
	note = {Number: 1
Publisher: Nature Publishing Group},
	pages = {6101},
}

@article{BrittonEtAl2020,
	title = {A mathematical model reveals the influence of population heterogeneity on herd immunity to {SARS}-{CoV}-2},
	volume = {369},
	url = {https://www.science.org/doi/10.1126/science.abc6810},
	doi = {10.1126/science.abc6810},
	number = {6505},
	urldate = {2022-08-07},
	journal = {Science},
	author = {Britton, Tom and Ball, Frank and Trapman, Pieter},
	month = aug,
	year = {2020},
	note = {Publisher: American Association for the Advancement of Science},
	pages = {846--849},
}

@article{GomesEtAl2022,
	title = {Individual variation in susceptibility or exposure to {SARS}-{CoV}-2 lowers the herd immunity threshold},
	volume = {540},
	issn = {0022-5193},
	url = {https://www.sciencedirect.com/science/article/pii/S0022519322000613},
	doi = {10.1016/j.jtbi.2022.111063},
	language = {en},
	urldate = {2022-08-07},
	journal = {Journal of Theoretical Biology},
	author = {Gomes, M. Gabriela M. and Ferreira, Marcelo U. and Corder, Rodrigo M. and King, Jessica G. and Souto-Maior, Caetano and Penha-Gonçalves, Carlos and Gonçalves, Guilherme and Chikina, Maria and Pegden, Wesley and Aguas, Ricardo},
	month = may,
	year = {2022},
	pages = {111063},
}

@article{OzEtAl2021,
	title = {Heterogeneity and superspreading effect on herd immunity},
	volume = {2021},
	issn = {1742-5468},
	url = {https://dx.doi.org/10.1088/1742-5468/abdfd1},
	doi = {10.1088/1742-5468/abdfd1},
	language = {en},
	number = {3},
	urldate = {2025-05-21},
	journal = {Journal of Statistical Mechanics: Theory and Experiment},
	author = {Oz, Yaron and Rubinstein, Ittai and Safra, Muli},
	month = mar,
	year = {2021},
	note = {Publisher: IOP Publishing and SISSA},
	pages = {033405},
}

@article{BarthelemyEtAl2005,
	title = {Dynamical patterns of epidemic outbreaks in complex heterogeneous networks},
	volume = {235},
	issn = {0022-5193},
	url = {https://www.sciencedirect.com/science/article/pii/S0022519305000251},
	doi = {10.1016/j.jtbi.2005.01.011},
	number = {2},
	urldate = {2025-05-21},
	journal = {Journal of Theoretical Biology},
	author = {Barthélemy, Marc and Barrat, Alain and Pastor-Satorras, Romualdo and Vespignani, Alessandro},
	month = jul,
	year = {2005},
	pages = {275--288},
}

@article{GoyalEtAl2021,
	title = {Viral load and contact heterogeneity predict {SARS}-{CoV}-2 transmission and super-spreading events},
	volume = {10},
	issn = {2050-084X},
	url = {https://doi.org/10.7554/eLife.63537},
	doi = {10.7554/eLife.63537},
	urldate = {2025-05-21},
	journal = {eLife},
	author = {Goyal, Ashish and Reeves, Daniel B and Cardozo-Ojeda, E Fabian and Schiffer, Joshua T and Mayer, Bryan T},
	editor = {Walczak, Aleksandra M and Childs, Lauren and Forde, Jonathan},
	month = feb,
	year = {2021},
	note = {Publisher: eLife Sciences Publications, Ltd},
	pages = {e63537},
}

@article{LauEtAl2020,
	title = {Characterizing superspreading events and age-specific infectiousness of {SARS}-{CoV}-2 transmission in {Georgia}, {USA}},
	volume = {117},
	url = {https://www.pnas.org/doi/full/10.1073/pnas.2011802117},
	doi = {10.1073/pnas.2011802117},
	number = {36},
	urldate = {2025-05-21},
	journal = {Proceedings of the National Academy of Sciences},
	author = {Lau, Max S. Y. and Grenfell, Bryan and Thomas, Michael and Bryan, Michael and Nelson, Kristin and Lopman, Ben},
	month = sep,
	year = {2020},
	note = {Publisher: Proceedings of the National Academy of Sciences},
	pages = {22430--22435},
}

@article{EndoEtAl2020,
	title = {Estimating the overdispersion in {COVID}-19 transmission using outbreak sizes outside {China}},
	url = {https://wellcomeopenresearch.org/articles/5-67},
	urldate = {2025-05-17},
	journal = {Wellcome Open Research},
	author = {Endo, Akira and Abbott, Sam and Kucharski, Adam J. and Funk, Sebastian},
	year = {2020},
}

@article{MillerEtAl2020,
	title = {Full genome viral sequences inform patterns of {SARS}-{CoV}-2 spread into and within {Israel}},
	volume = {11},
	copyright = {2020 The Author(s)},
	issn = {2041-1723},
	url = {https://www.nature.com/articles/s41467-020-19248-0},
	doi = {10.1038/s41467-020-19248-0},
	language = {en},
	number = {1},
	urldate = {2025-05-21},
	journal = {Nature Communications},
	author = {Miller, Danielle and Martin, Michael A. and Harel, Noam and Tirosh, Omer and Kustin, Talia and Meir, Moran and Sorek, Nadav and Gefen-Halevi, Shiraz and Amit, Sharon and Vorontsov, Olesya and Shaag, Avraham and Wolf, Dana and Peretz, Avi and Shemer-Avni, Yonat and Roif-Kaminsky, Diana and Kopelman, Naama M. and Huppert, Amit and Koelle, Katia and Stern, Adi},
	month = nov,
	year = {2020},
	note = {Publisher: Nature Publishing Group},
	pages = {5518},
}

@article{SneppenEtAl2021,
	title = {Overdispersion in {COVID}-19 increases the effectiveness of limiting nonrepetitive contacts for transmission control},
	volume = {118},
	url = {https://www.pnas.org/doi/full/10.1073/pnas.2016623118},
	doi = {10.1073/pnas.2016623118},
	number = {14},
	urldate = {2025-05-20},
	journal = {Proceedings of the National Academy of Sciences},
	author = {Sneppen, Kim and Nielsen, Bjarke Frost and Taylor, Robert J. and Simonsen, Lone},
	month = apr,
	year = {2021},
	note = {Publisher: Proceedings of the National Academy of Sciences},
	pages = {e2016623118},
}

@article{BallNeal2008,
	title = {Network epidemic models with two levels of mixing},
	volume = {212},
	issn = {0025-5564},
	url = {https://www.sciencedirect.com/science/article/pii/S0025556408000023},
	doi = {10.1016/j.mbs.2008.01.001},
	number = {1},
	urldate = {2025-05-18},
	journal = {Mathematical Biosciences},
	author = {Ball, Frank and Neal, Peter},
	month = mar,
	year = {2008},
	pages = {69--87},
}

@article{Volz2008,
	title = {{SIR} dynamics in random networks with heterogeneous connectivity},
	volume = {56},
	issn = {1432-1416},
	url = {https://doi.org/10.1007/s00285-007-0116-4},
	doi = {10.1007/s00285-007-0116-4},
	language = {en},
	number = {3},
	urldate = {2025-05-18},
	journal = {Journal of Mathematical Biology},
	author = {Volz, Erik},
	month = mar,
	year = {2008},
	pages = {293--310},
}

@article{KermackMcKendrick1927,
	title = {A contribution to the mathematical theory of epidemics},
	volume = {115},
	url = {https://royalsocietypublishing.org/doi/10.1098/rspa.1927.0118},
	doi = {10.1098/rspa.1927.0118},
	number = {772},
	urldate = {2022-08-24},
	journal = {Proceedings of the Royal Society of London. Series A, Containing Papers of a Mathematical and Physical Character},
	author = {Kermack, William Ogilvy and McKendrick, A. G.},
	month = aug,
	year = {1927},
	note = {Publisher: Royal Society},
	pages = {700--721},
}

@incollection{Heesterbeek2005,
	address = {Burlington},
	series = {Theoretical {Ecology} {Series}},
	title = {5 - {THE} {LAW} {OF} {MASS}-{ACTION} {IN} {EPIDEMIOLOGY}: {A} {HISTORICAL} {PERSPECTIVE}},
	shorttitle = {5 - {THE} {LAW} {OF} {MASS}-{ACTION} {IN} {EPIDEMIOLOGY}},
	url = {https://www.sciencedirect.com/science/article/pii/B9780120884599500078},
	urldate = {2025-05-04},
	booktitle = {Ecological {Paradigms} {Lost}},
	publisher = {Academic Press},
	author = {Heesterbeek, Hans},
	editor = {Cuddington, Kim and Beisner, Beatrix E.},
	month = jan,
	year = {2005},
	doi = {10.1016/B978-012088459-9/50007-8},
	pages = {81--105},
}

@article{LindquistEtAl2011,
	title = {Effective degree network disease models},
	volume = {62},
	issn = {1432-1416},
	url = {https://doi.org/10.1007/s00285-010-0331-2},
	doi = {10.1007/s00285-010-0331-2},
	language = {en},
	number = {2},
	urldate = {2025-05-18},
	journal = {Journal of Mathematical Biology},
	author = {Lindquist, Jennifer and Ma, Junling and van den Driessche, P. and Willeboordse, Frederick H.},
	month = feb,
	year = {2011},
	pages = {143--164},
}

@article{Bauch2002,
	title = {A versatile {ODE} approximation to a network model for the spread of sexually transmitted diseases},
	volume = {45},
	issn = {1432-1416},
	url = {https://doi.org/10.1007/s002850200153},
	doi = {10.1007/s002850200153},
	language = {en},
	number = {5},
	urldate = {2025-05-18},
	journal = {Journal of Mathematical Biology},
	author = {Bauch, C.T.},
	month = nov,
	year = {2002},
	pages = {375--395},
}

@article{EamesKeeling2002,
	title = {Modeling dynamic and network heterogeneities in the spread of sexually transmitted diseases},
	volume = {99},
	url = {https://www.pnas.org/doi/10.1073/pnas.202244299},
	doi = {10.1073/pnas.202244299},
	number = {20},
	urldate = {2025-05-18},
	journal = {Proceedings of the National Academy of Sciences},
	author = {Eames, Ken T. D. and Keeling, Matt J.},
	month = oct,
	year = {2002},
	note = {Publisher: Proceedings of the National Academy of Sciences},
	pages = {13330--13335},
}

@article{BalcanEtAl2009,
	title = {Multiscale mobility networks and the spatial spreading of infectious diseases},
	volume = {106},
	url = {https://www.pnas.org/doi/full/10.1073/pnas.0906910106},
	doi = {10.1073/pnas.0906910106},
	number = {51},
	urldate = {2025-05-18},
	journal = {Proceedings of the National Academy of Sciences},
	author = {Balcan, Duygu and Colizza, Vittoria and Gonçalves, Bruno and Hu, Hao and Ramasco, José J. and Vespignani, Alessandro},
	month = dec,
	year = {2009},
	note = {Publisher: Proceedings of the National Academy of Sciences},
	pages = {21484--21489},
}

@article{EubankEtAl2004,
	title = {Modelling disease outbreaks in realistic urban social networks},
	volume = {429},
	copyright = {2004 Macmillan Magazines Ltd.},
	issn = {1476-4687},
	url = {https://www.nature.com/articles/nature02541},
	doi = {10.1038/nature02541},
	language = {en},
	number = {6988},
	urldate = {2025-05-18},
	journal = {Nature},
	author = {Eubank, Stephen and Guclu, Hasan and Anil Kumar, V. S. and Marathe, Madhav V. and Srinivasan, Aravind and Toroczkai, Zoltán and Wang, Nan},
	month = may,
	year = {2004},
	note = {Publisher: Nature Publishing Group},
	pages = {180--184},
}

@article{LloydSmithEtAl2005,
	title = {Superspreading and the effect of individual variation on disease emergence},
	volume = {438},
	issn = {14764687},
	url = {https://www.nature.com/articles/nature04153},
	doi = {10.1038/nature04153},
	number = {7066},
	journal = {Nature},
	author = {Lloyd-Smith, J. O. and Schreiber, S. J. and Kopp, P. E. and Getz, W. M.},
	month = nov,
	year = {2005},
	pmid = {16292310},
	note = {Publisher: Nature Publishing Group},
	pages = {355--359},
}

@article{MiyamaEtAl2022,
	title = {Decrease in overdispersed secondary transmission of {COVID}-19 over time in {Japan}},
	volume = {150},
	issn = {0950-2688, 1469-4409},
	url = {https://www.cambridge.org/core/journals/epidemiology-and-infection/article/decrease-in-overdispersed-secondary-transmission-of-covid19-over-time-in-japan/B9D3F2DE0E43C547204214F1E582DFBE},
	doi = {10.1017/S0950268822001789},
	language = {en},
	urldate = {2024-10-08},
	journal = {Epidemiology \& Infection},
	author = {Miyama, Takeshi and Jung, Sung-mok and Nishiura, Hiroshi},
	month = jan,
	year = {2022},
	pages = {e197},
}

@article{Miller2011,
	title = {A note on a paper by {Erik} {Volz}: {SIR} dynamics in random networks},
	volume = {62},
	issn = {03036812},
	url = {https://link.springer.com/article/10.1007/s00285-010-0337-9},
	doi = {10.1007/s00285-010-0337-9},
	number = {3},
	journal = {Journal of Mathematical Biology},
	author = {Miller, Joel C.},
	month = mar,
	year = {2011},
	pmid = {20309549},
	note = {Publisher: Springer},
	pages = {349--358},
}

@article{MillerEtAl2012,
	title = {Edge-based compartmental modelling for infectious disease spread},
	volume = {9},
	issn = {1742-5689},
	url = {https://royalsocietypublishing.org/doi/10.1098/rsif.2011.0403},
	doi = {10.1098/rsif.2011.0403},
	number = {70},
	journal = {Journal of The Royal Society Interface},
	author = {Miller, Joel C. and Slim, Anja C. and Volz, Erik M.},
	month = may,
	year = {2012},
	note = {Publisher: Royal Society},
	pages = {890--906},
}

@article{GuoEtAl2023,
	title = {A statistical framework for tracking the time-varying superspreading potential of {COVID}-19 epidemic},
	volume = {42},
	issn = {1755-4365},
	url = {https://www.sciencedirect.com/science/article/pii/S1755436523000063},
	doi = {10.1016/j.epidem.2023.100670},
	urldate = {2025-05-17},
	journal = {Epidemics},
	author = {Guo, Zihao and Zhao, Shi and Lee, Shui Shan and Hung, Chi Tim and Wong, Ngai Sze and Chow, Tsz Yu and Yam, Carrie Ho Kwan and Wang, Maggie Haitian and Wang, Jingxuan and Chong, Ka Chun and Yeoh, Eng Kiong},
	month = mar,
	year = {2023},
	pages = {100670},
}

@article{BrainardEtAl2023,
	title = {Super-spreaders of novel coronaviruses that cause {SARS}, {MERS} and {COVID}-19: a systematic review},
	volume = {82},
	issn = {1047-2797},
	shorttitle = {Super-spreaders of novel coronaviruses that cause {SARS}, {MERS} and {COVID}-19},
	url = {https://www.sciencedirect.com/science/article/pii/S1047279723000583},
	doi = {10.1016/j.annepidem.2023.03.009},
	urldate = {2025-05-17},
	journal = {Annals of Epidemiology},
	author = {Brainard, Julii and Jones, Natalia R. and Harrison, Florence C. D. and Hammer, Charlotte C. and Lake, Iain R.},
	month = jun,
	year = {2023},
	pages = {66--76.e6},
}

@article{NielsenEtAl2021,
	title = {{COVID}-19 {Superspreading} {Suggests} {Mitigation} by {Social} {Network} {Modulation}},
	volume = {126},
	url = {https://link.aps.org/doi/10.1103/PhysRevLett.126.118301},
	doi = {10.1103/PhysRevLett.126.118301},
	number = {11},
	urldate = {2025-05-11},
	journal = {Physical Review Letters},
	author = {Nielsen, Bjarke Frost and Simonsen, Lone and Sneppen, Kim},
	month = mar,
	year = {2021},
	note = {Publisher: American Physical Society},
	pages = {118301},
}

@article{NielsenEtAl2023,
	title = {Host heterogeneity and epistasis explain punctuated evolution of {SARS}-{CoV}-2},
	volume = {19},
	issn = {1553-7358},
	url = {https://journals.plos.org/ploscompbiol/article?id=10.1371/journal.pcbi.1010896},
	doi = {10.1371/journal.pcbi.1010896},
	language = {en},
	number = {2},
	urldate = {2025-05-11},
	journal = {PLOS Computational Biology},
	author = {Nielsen, Bjarke Frost and Saad-Roy, Chadi M. and Li, Yimei and Sneppen, Kim and Simonsen, Lone and Viboud, Cécile and Levin, Simon A. and Grenfell, Bryan T.},
	month = feb,
	year = {2023},
	note = {Publisher: Public Library of Science},
	pages = {e1010896},
}

@book{BarratEtAl2013,
	address = {Cambridge, UK},
	title = {Dynamical {Processes} on {Complex} {Networks}},
	volume = {37},
	shorttitle = {Dynamical {Processes} on {Complex} {Networks} (4th ed.) by {A}. {Barrat}, {M}. {Barthélemy}, \& {A}. {Vespignani}},
	url = {https://doi.org/10.1080/0022250X.2012.728886},
	urldate = {2025-05-18},
	publisher = {Cambridge University},
	author = {Barrat, A. and Barthélemy, M. and Vespignani, A.},
	month = apr,
	year = {2013},
	note = {Publisher: Routledge
\_eprint: https://doi.org/10.1080/0022250X.2012.728886},
}

@article{Klovdahl1985,
	title = {Social networks and the spread of infectious diseases: {The} {AIDS} example},
	volume = {21},
	issn = {0277-9536},
	shorttitle = {Social networks and the spread of infectious diseases},
	url = {https://www.sciencedirect.com/science/article/pii/0277953685902692},
	doi = {10.1016/0277-9536(85)90269-2},
	number = {11},
	urldate = {2025-05-18},
	journal = {Social Science \& Medicine},
	author = {Klovdahl, Alden S.},
	month = jan,
	year = {1985},
	pages = {1203--1216},
}

@article{WattsStrogatz1998,
	title = {Collective dynamics of ‘small-world’ networks},
	volume = {393},
	copyright = {1998 Macmillan Magazines Ltd.},
	issn = {1476-4687},
	url = {https://www.nature.com/articles/30918},
	doi = {10.1038/30918},
	language = {en},
	number = {6684},
	urldate = {2025-05-18},
	journal = {Nature},
	author = {Watts, Duncan J. and Strogatz, Steven H.},
	month = jun,
	year = {1998},
	note = {Publisher: Nature Publishing Group},
	pages = {440--442},
}

@article{Newman2002,
	title = {Spread of epidemic disease on networks},
	volume = {66},
	url = {https://link.aps.org/doi/10.1103/PhysRevE.66.016128},
	doi = {10.1103/PhysRevE.66.016128},
	number = {1},
	urldate = {2025-05-18},
	journal = {Physical Review E},
	author = {Newman, M. E. J.},
	month = jul,
	year = {2002},
	note = {Publisher: American Physical Society},
	pages = {016128},
}

@article{NoelEtAl2009,
	title = {Time evolution of epidemic disease on finite and infinite networks},
	volume = {79},
	url = {https://link.aps.org/doi/10.1103/PhysRevE.79.026101},
	doi = {10.1103/PhysRevE.79.026101},
	number = {2},
	urldate = {2025-05-17},
	journal = {Physical Review E},
	author = {Noël, Pierre-André and Davoudi, Bahman and Brunham, Robert C. and Dubé, Louis J. and Pourbohloul, Babak},
	month = feb,
	year = {2009},
	note = {Publisher: American Physical Society},
	pages = {026101},
}

\setcounter{section}{0} 
\renewcommand{\thesection}{\arabic{section}} 
\setcounter{figure}{0} 
\renewcommand{\thefigure}{S\arabic{figure}} 
\setcounter{table}{0} 
\renewcommand{\thetable}{S\arabic{table}} 
\setcounter{theorem}{0}
\renewcommand{\thetheorem}{S\arabic{theorem}}
\setcounter{equation}{0}
\renewcommand{\theequation}{S\arabic{equation}}

\title{
\vspace{-1cm}
Supplement to:\\Tracking dynamics of superspreading through contacts, exposures, and transmissions in edge-based network epidemics}
\author{Ari S. Freedman$^{1,2,3,*}$, Bjarke F. Nielsen$^{4,5,6}$, Maximillian M. Nguyen$^{7,8}$,\\Laurent H\'ebert-Dufresne$^{3,9,10}$, Simon A. Levin$^1$}
\date{
\small
$^1$Department of Ecology and Evolutionary Biology, Princeton University, Princeton, NJ, USA\\
$^2$Department of Plant Biology, University of Vermont, Burlington, VT, USA\\
$^3$Vermont Complex Systems Institute, University of Vermont, Burlington, VT, USA\\
$^4$High Meadows Environmental Institute, Princeton University, Princeton, NJ, USA\\
$^5$Niels Bohr Institute, University of Copenhagen, Copenhagen, Denmark\\
$^6$PandemiX Center, Roskilde University, Roskilde, Denmark\\
$^7$Lewis-Sigler Institute, Princeton University, Princeton, NJ, USA\\
$^8$School of Medicine, Emory University, Atlanta, GA, USA\\
$^9$Complexity Science Hub, Vienna, Austria\\
$^{10}$Santa Fe Institute, Santa Fe, NM, USA\\
$^*$Corresponding author; email: ari.freedman@uvm.edu
}

\onehalfspacing

\emergencystretch 3em

\maketitle

\tableofcontents

\clearpage

\section{Supplementary figures}

\begin{figure}[!ht]
\centering
\includegraphics[width=\textwidth]{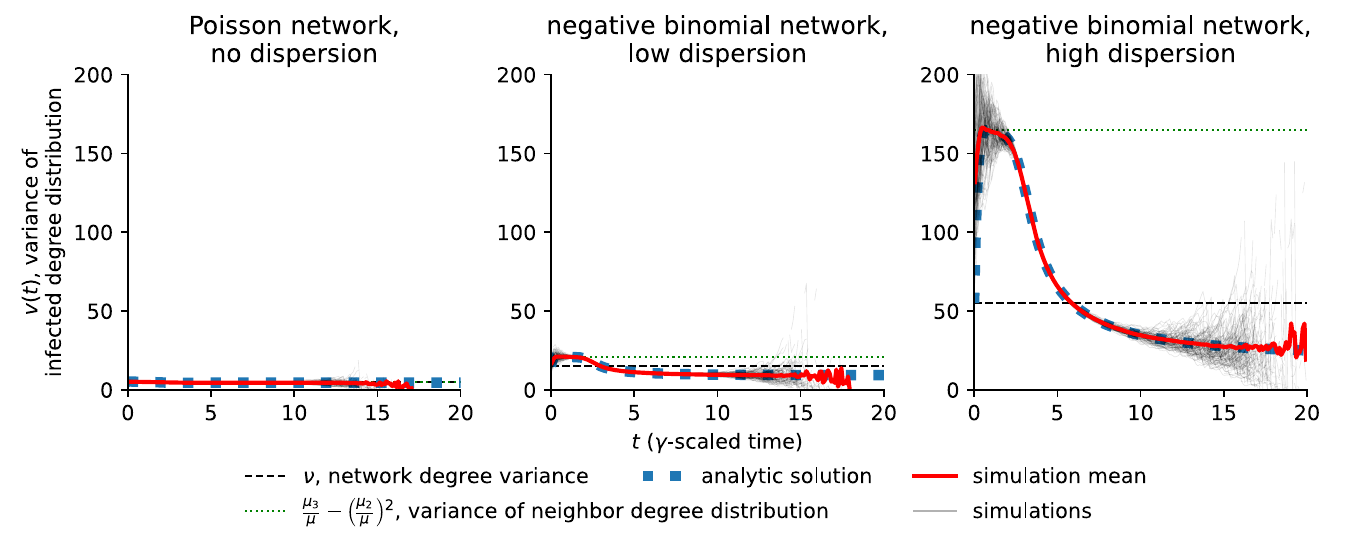}
\caption{Trajectories of the infected degree distribution's variance $v(t)$ for three different network degree distributions. Red curves show the mean of 200 simulations (whose individual trajectories are also plotted by faint gray curves) while the blue dotted curves show the analytic solution from Theorem 1 in the main text. Black dashed lines show the variance $\nu$ of the network degree distribution, and green dashed lines show the variance $\frac{\mu_3}{\mu_1}-(\frac{\mu_2}{\mu_1})^2$ of the neighbor degree distribution (where $\mu_n$ is the $n$-th moment of the network degree distribution).}
\label{fig:vs}
\end{figure}

\clearpage
\begin{figure}[!ht]
\centering
\includegraphics[width=\textwidth]{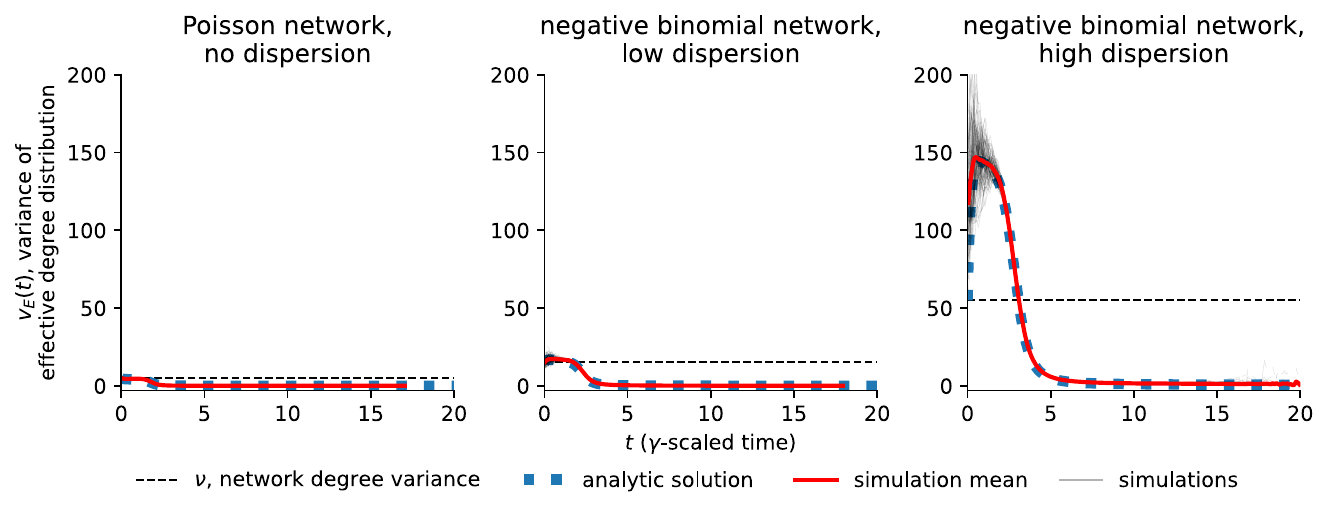}
\caption{Trajectories of the effective degree distribution's variance $v_E(t)$ for three different network degree distributions. Red curves show the mean of 200 simulations (whose individual trajectories are also plotted by faint gray curves) while the blue dotted curves show the analytic solution from Eq.~\eqref{eq:supplement:mE2}. Black dashed lines show the variance  $\nu$ of the network degree distribution.}
\label{fig:vEs}
\end{figure}

\clearpage
\begin{figure}[!ht]
\centering
\includegraphics[width=\textwidth]{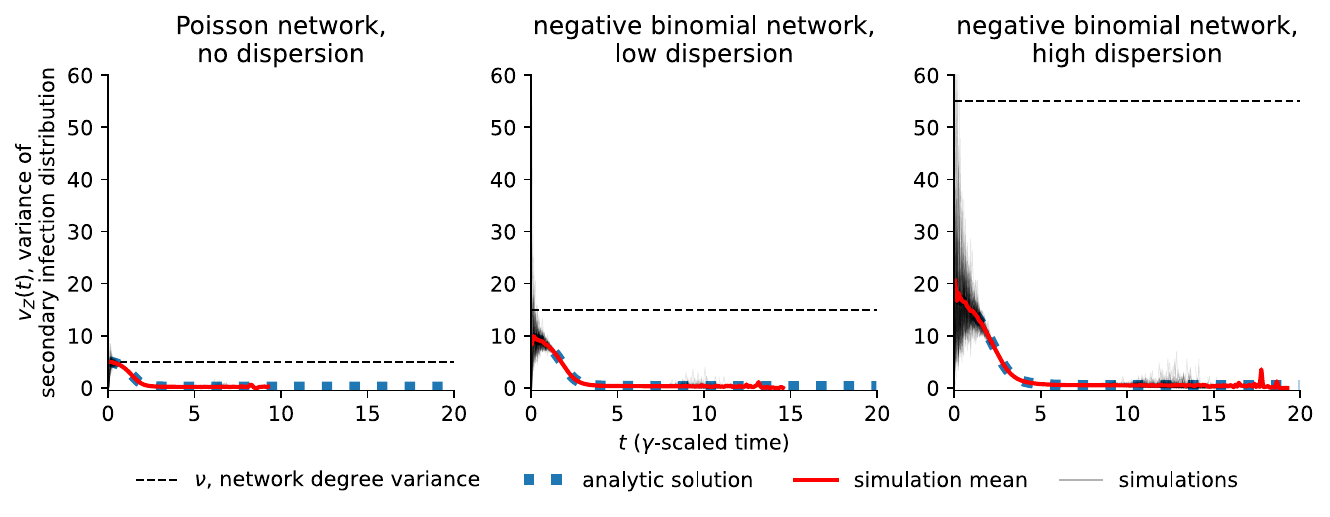}
\caption{Trajectories of the effective degree distribution's variance $v_Z(t)$ for three different network degree distributions. Red curves show the mean of 200 simulations (whose individual trajectories are also plotted by faint gray curves) while the blue dotted curves show the analytic solution from Eq.~\eqref{eq:supplement:mZ2}. Black dashed lines show the variance $\nu$ of the network degree distribution.}
\label{fig:vZs}
\end{figure}

\clearpage
\begin{figure}[!ht]
\centering
\includegraphics[width=\textwidth]{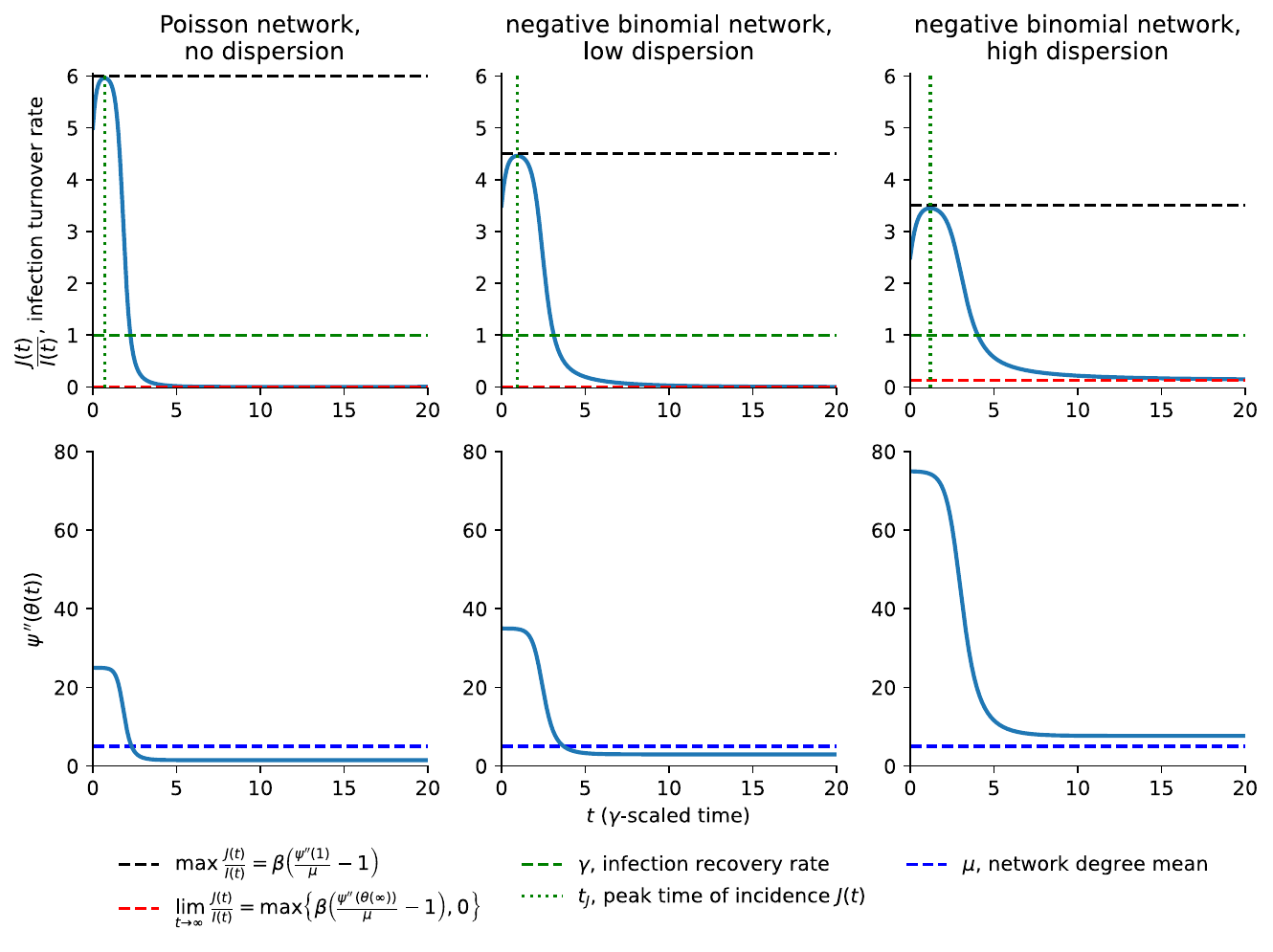}
\caption{Trajectories of the infection turnover rate $\frac{J(t)}{I(t)}$ (top row) and $\psi''(\theta(t))$ for three different network degree distributions, summarizing the results of Lemma 2 in the main text. Black dashed lines show $\beta\big(\frac{\psi''(1)}{\mu}-1\big)$, the peak value $\frac{J(t)}{I(t)}$ tends towards during the beginning of an epidemic. Green dashed lines show the recovery rate $\gamma$, a lower bound to $\frac{J(t)}{I(t)}$ for all $t$ less than the the peak time $t_J$ (green dotted lines) of incidence $J(t)$. Red dashed lines show $\max\big\{\beta\big(\frac{\psi''(\theta(\infty))}{\mu}-1),0\big\}$, the limit of $\frac{J(t)}{I(t)}$ as $t\to\infty$. If $\psi''(\theta(\infty))\le\mu$ (which occurs for the ``no dispersion'' and ``low dispersion'' networks, as shown by the bottom row of plots), then $\frac{J(t)}{I(t)}$ will tend vanish to 0 as $t\to\infty$, whereas if $\psi''(\theta(\infty))>\mu$ (which occurs for the ``high dispersion'' network) then the limit of $\frac{J(t)}{I(t)}$ as $t\to\infty$ will be greater than 0.}
\label{fig:J_I}
\end{figure}

\clearpage
\section{Deriving the second moments of \texorpdfstring{$E(t)$ and $Z(t)$}{E(t) and Z(t)}}

\begin{theorem}
The second moment $m_{E,2}(t)$ of the effective degree distribution $E(t)$ is given by
\begin{equation}
\label{eq:supplement:mE2}
\begin{aligned}
m_{E,2}&=m_E(t)+\frac{e^{-(2\beta+\gamma)t}I(0)}{I(t)}\left(\frac{\psi'(\theta(t))}{\psi'(1)}\right)^2\psi''(1)\\
&\qquad+\frac1{I(t)}\left(\frac{\psi'(\theta(t))}{\psi'(1)}\right)^2\sum_{j=1}^\infty P(j)(j-1)(j-2)\frac{\left(H_j(t)-e^{-(\beta+\gamma)t}I(0)\right)^2}{I_j(t)-e^{-\gamma t}I(0)}.
\end{aligned}
\end{equation}
\label{thm:supplement:mE2}
\end{theorem}
\begin{proof}
We repeat for convenience the mass function for $E(t)$ from the main text:
\begin{equation}
p_{E,k}(t)=\sum_{j=k}^\infty\frac{P(j)e^{-\gamma t}I(0)}{I(t)}\eta_{1,j,k}(t)+\sum_{j=k+1}^\infty\frac{P(j)(I_j(t)-e^{-\gamma t}I(0))}{I(t)}\eta_{2,j,k}(t)
\label{eq:supplement:pEk}
\end{equation}
with
\begin{equation}
\eta_{1,j,k}(t)=\binom jk\left(\frac{\psi'(\theta(t))}{\psi'(1)}e^{-\beta t}\right)^k\left(1-\frac{\psi'(\theta(t))}{\psi'(1)}e^{-\beta t}\right)^{j-k}.
\label{eq:supplement:eta1}
\end{equation}
and
\begin{equation}
\begin{aligned}
\eta_{2,j,k}(t)=&\binom{j-1}k\left(\frac{\psi'(\theta(t))}{\psi'(1)}\frac{H_j(t)-e^{-(\beta+\gamma)t}I(0)}{I_j(t)-e^{-\gamma t}I(0)}\right)^k\times\\
&\quad\left(1-\frac{\psi'(\theta(t))}{\psi'(1)}\frac{H_j(t)-e^{-(\beta+\gamma)t}I(0)}{I_j(t)-e^{-\gamma t}I(0)}\right)^{j-1-k}.
\end{aligned}
\label{eq:supplement:eta2}
\end{equation}
By the second moment formula for binomial distributions
\begin{equation}
\label{eq:supplement:eta1_second_moment}
\sum_{k=0}^{j}k^2\eta_{1,j,k}(t)=j\left(\frac{\psi'(\theta(t))}{\psi'(1)}e^{-\beta t}\right)+j(j-1)\left(\frac{\psi'(\theta(t))}{\psi'(1)}e^{-\beta t}\right)^2.
\end{equation}
and
\begin{equation}
\label{eq:supplement:eta2_second_moment}
\begin{aligned}
\sum_{k=0}^{j-1}k^2\eta_{2,j,k}(t)&=(j-1)\left(\frac{\psi'(\theta(t))}{\psi'(1)}\frac{H_j(t)-e^{-(\beta+\gamma)t}I(0)}{I_j(t)-e^{-\gamma t}I(0)}\right)+\\
&\qquad(j-1)(j-2)\left(\frac{\psi'(\theta(t))}{\psi'(1)}\frac{H_j(t)-e^{-(\beta+\gamma)t}I(0)}{I_j(t)-e^{-\gamma t}I(0)}\right)^2.
\end{aligned}
\end{equation}
Now we seek
\begin{equation}
m_{E,2}(t)=\sum_{k=0}^\infty\sum_{j=k}^\infty\frac{P(j)e^{-\gamma t}I(0)}{I(t)}k^2\eta_{1,j,k}(t)+\sum_{k=0}^\infty\sum_{j=k+1}^\infty\frac{P(j)(I_j(t)-e^{-\gamma t}I(0))}{I(t)}k^2\eta_{2,j,k}(t)
\label{eq:supplement:mE2_0}
\end{equation}
Putting the first terms of Eq.~\eqref{eq:supplement:eta1_second_moment} and Eq.~\eqref{eq:supplement:eta2_second_moment} into the sums of Eq.~\eqref{eq:supplement:mE2_0} will yield $m_E(t)$, in the same way that is calculated in the main text. Then putting just the second terms of Eq.~\eqref{eq:supplement:eta1_second_moment} and Eq.~\eqref{eq:supplement:eta2_second_moment} into the sums of Eq.~\eqref{eq:supplement:mE2_0} yields
\begin{align*}
&m_{E,2}(t)-m_E(t)=\sum_{j=0}^\infty\frac{P(j)e^{-\gamma t}I(0)}{I(t)}j(j-1)\left(\frac{\psi'(\theta(t))}{\psi'(1)}e^{-\beta t}\right)^2+\\
&\qquad\qquad\sum_{j=1}^\infty\frac{P(j)(I_j(t)-e^{-\gamma t}I(0))}{I(t)}(j-1)(j-2)\left(\frac{\psi'(\theta(t))}{\psi'(1)}\frac{H_j(t)-e^{-(\beta+\gamma)t}I(0)}{I_j(t)-e^{-\gamma t}I(0)}\right)^2.
\end{align*}
This quickly gives the desired result, noting that $\psi''(1)=\sum_{j=0}^\infty P(j)j(j-1)$.
\end{proof}

\begin{theorem}
The second moment $m_{Z,2}(t)$ to the secondary infection distribution $Z(t)$ is given by
\begin{equation}
m_{Z,2}(t)=m_Z(t)+\frac{\psi'''(\theta(t))}{\psi'(\theta(t))}\int_t^\infty\gamma e^{-\gamma(s-t)}\left(\int_t^s\beta e^{-\beta(\tau-t)}\frac{\psi'(\theta(\tau))}{\psi'(1)}d\tau\right)^2ds.
\label{eq:supplement:mZ2}
\end{equation}
\label{thm:supplement:mZ2}
\end{theorem}
\begin{proof}
We repeat for convenience the mass function for $Z(t)$ from the main text:
\begin{equation}
p_{Z,k}(t)=\sum_{j=k+1}^\infty\frac{j\theta(t)^{j-1}}{\psi'(\theta(t))}\int_t^\infty\gamma e^{-\gamma(s-t)}\binom{j-1}{k}\zeta_t(s)^k(1-\zeta_t(s))^{j-1-k}ds,
\label{eq:supplement:pZk}
\end{equation}
where
\begin{equation}
\zeta_t(s)=\sum_{k=0}^\infty\frac{P(k)k}{\mu}\zeta_t^k(s)=\frac1{\theta(t)}\int_t^s\beta e^{-\beta(\tau-t)}\frac{\psi'(\theta(\tau))}{\psi'(1)}\ d\tau.
\label{eq:supplement:zeta}
\end{equation}
Now we seek
\begin{align*}
m_{Z,2}(t)&=\sum_{k=0}^\infty k^2\sum_{j=k+1}^\infty\frac{j\theta(t)^{j-1}}{\psi'(\theta(t))}\int_t^\infty\gamma e^{-\gamma(s-t)}\binom{j-1}{k}\zeta_t(s)^k(1-\zeta_t(s))^{j-1-k}ds\\
&=\sum_{j=1}^\infty\frac{j\theta(t)^{j-1}}{\psi'(\theta(t))}\int_t^\infty\gamma e^{-\gamma(s-t)}\sum_{k=0}^{j-1}k^2\binom{j-1}{k}\zeta_t(s)^k(1-\zeta_t(s))^{j-1-k}ds\\
&=\sum_{j=1}^\infty\frac{j\theta(t)^{j-1}}{\psi'(\theta(t))}\int_t^\infty\gamma e^{-\gamma(s-t)}\left((j-1)\zeta_t(s)+(j-1)(j-2)\zeta_t(s)^2\right)ds\\
&=m_Z(t)+\sum_{j=1}^\infty\frac{j(j-1)(j-2)\theta(t)^{j-1}}{\psi'(\theta(t))}\int_t^\infty\gamma e^{-\gamma(s-t)}\zeta_t(s)^2ds,
\end{align*}
where we applied the formula for the second moment of a binomial distribution and pulled out $m_Z(t)$ by recognizing the same calculations used to derive $m_Z(t)$ in the main text. The desired result then immediately follows.
\end{proof}

\clearpage
\section{Additional distributions for newly infected nodes}

Of the three distributions we describe in the main text for the system at time $t$, two of the distributions (the infected degree distribution $X(t)$ and the effective degree distribution $E(t)$) are concerned with the properties of \textit{all currently infected nodes} at time $t$, while the third distribution (the secondary infection distribution $Z(t)$) is concerned with the properties of \textit{newly infected nodes} at time $t$. We mainly focus on the properties of all currently infected nodes when considering the potential for superspreading, since nodes that were infected long previously can still have superspreading potential if they continue to be infected and have many susceptible neighbors. For the secondary infection distribution, however, we consider newly infected nodes since it only makes sense to count the total number of secondary cases an infected node will cause if the count begins right when the node is infected.

For the sake of completion, we describe here two additional distributions that are variations of the infected degree distribution $X(t)$ and the effective degree distribution $E(t)$, except looking at the only newly infected nodes at time $t$. The first is the ``newly infected degree distribution'' $Y(t)$, which is the distribution of the degrees of newly infected nodes at time $t$; this is the same as the infected degree distribution but restricted to newly infected nodes. And the second is the ``effective newly infected degree distribution'' $F(t)$, which is the distribution of the number of susceptible neighbors that nodes newly infected at time $t$ have at time $t$; this is the same as the effective degree distribution but restricted to newly infected nodes. By definition $F(t)\le Y(t)$ always, just as $E(t)\le X(t)$ always.

\begin{theorem}
The mass function $p_{Y,k}(t)$ of the newly infected degree distribution $Y(t)$ is given by
\begin{equation}
\label{eq:supplement:pYk}
p_{Y,k}=\frac{P(k)k\theta^{k-1}}{\psi'(\theta)}=\frac{P(k)k\theta^k}{\phi'(\log\theta)},
\end{equation}
and its $n$-th moment $m_{Y,n}(t)$ by
\begin{equation}
\label{eq:supplement:mYn}
m_{Y,n}=\frac{\phi^{n+1}(\log\theta)}{\phi'(\log\theta)}
\end{equation}
for any $n$ for which $\mu_{n+1}$, the $(n+1)$-th moment of the network degree distribution, exists. In particular, $Y(t)$ has mean
\begin{equation}
\label{eq:supplement:mY}
m_Y=m_{Y,1}=\frac{\phi''(\log\theta)}{\phi'(\log\theta)}
\end{equation}
and second moment
\begin{equation}
\label{eq:supplement:mY2}
m_{Y,2}=\frac{\phi'''(\log\theta)}{\phi'(\log\theta)}.
\end{equation}
\end{theorem}
\begin{proof}
In the main text, we established that the instantaneous incidence of new infections of degree $k$ is given by $J_k=-\dot S_k=-k\theta^{k-1}\dot\theta$. Then by Bayes' law and the fact that $\phi'(\log\theta)=\psi'(\theta)\theta$,
\begin{equation*}
p_{Y,k}=\frac{P(k)J_k}{J}=\frac{P(k)k\theta^{k-1}}{\psi'(\theta)}=\frac{P(k)k\theta^k}{\phi'(\log\theta)}.
\end{equation*}
We then get that the $n$-th moment of $Y(t)$ is given by
\begin{equation*}
m_{Y,n}=\sum_{k=0}^\infty p_{Y,k}k^n=\frac{P(k)k^{n+1}\theta^k}{\phi'(\log\theta)}=\frac{\phi^{n+1}(\log\theta)}{\phi'(\log\theta)},
\end{equation*}
where the existence of $\mu_{n+1}$ ensures that $\phi^{(n+1)}(\log\theta)=\langle K^{n+1}\theta^K\rangle\le\langle K^{n+1}\rangle=\mu_{n+1}$ exists, with $K$ being the network's degree distribution.
\end{proof}

Note that $p_{Y,k}=\frac{P(k)k\theta^k}{\phi'(\log\theta)}$ is the moving target that $p_k(t)$, the mass function of the infected degree distribution, is tending towards in $\dot p_k=-\frac{J}{I}\left(p_k-\frac{P(k)k\theta^k}{\phi'(\log\theta)}\right)$; and $m_{Y,n}=\frac{\phi^{n+1}(\log\theta)}{\phi'(\log\theta)}$ is the moving target that $m_n(t)$, the $n$-th moment of the infected degree distribution, is tending towards in $\dot m_n=-\frac{J}{I}\left(m_n-\frac{\phi^{(n+1)}(\log\theta)}{\phi'(\log\theta)}\right)$. Essentially, the infected degree distribution $X(t)$ is always tending towards the newly infected degree distribution $Y(t)$, at the rate $\frac JI$ we call the infection turnover rate.

The mean and variance of the newly infected degree distribution $Y(t)$ are shown in Figs.~\ref{fig:mYs}--\ref{fig:vYs}, respectively.

\clearpage

\begin{figure}[!t]
\centering
\includegraphics[width=\textwidth]{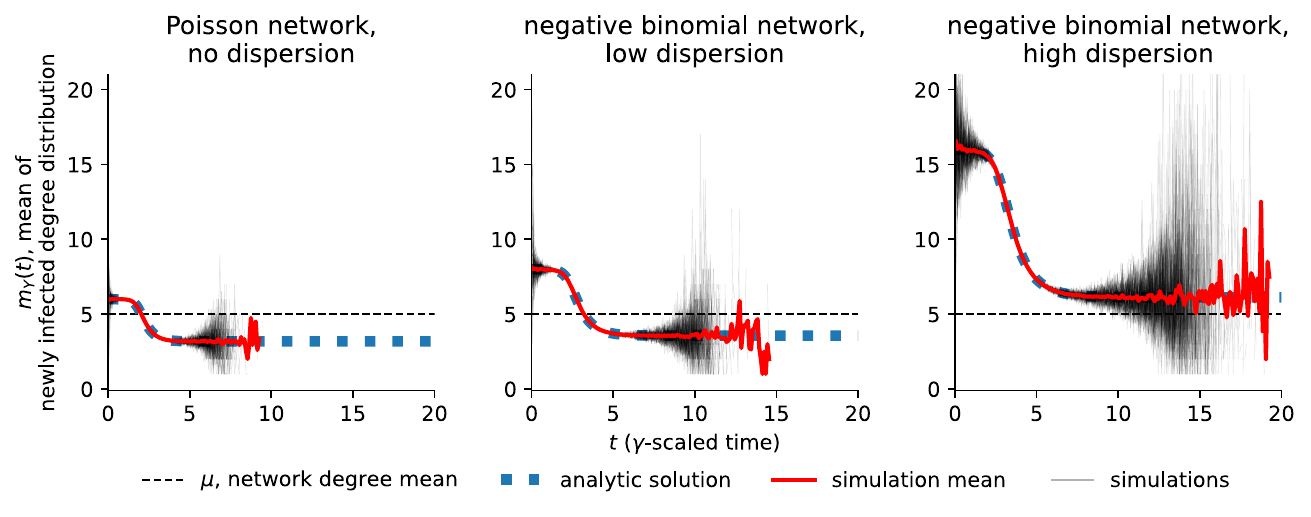}
\caption{Trajectories of the newly infected degree distribution's mean $m_Y(t)$ for three different network degree distributions. Red curves show the mean of 200 simulations (whose individual trajectories are also plotted by faint gray curves) while the blue dotted curves show the analytic solution from Eq.~\eqref{eq:supplement:mY}. Black dashed lines show the mean $\mu$ of the network degree distribution.}
\label{fig:mYs}
\end{figure}

\begin{figure}[!b]
\centering
\includegraphics[width=\textwidth]{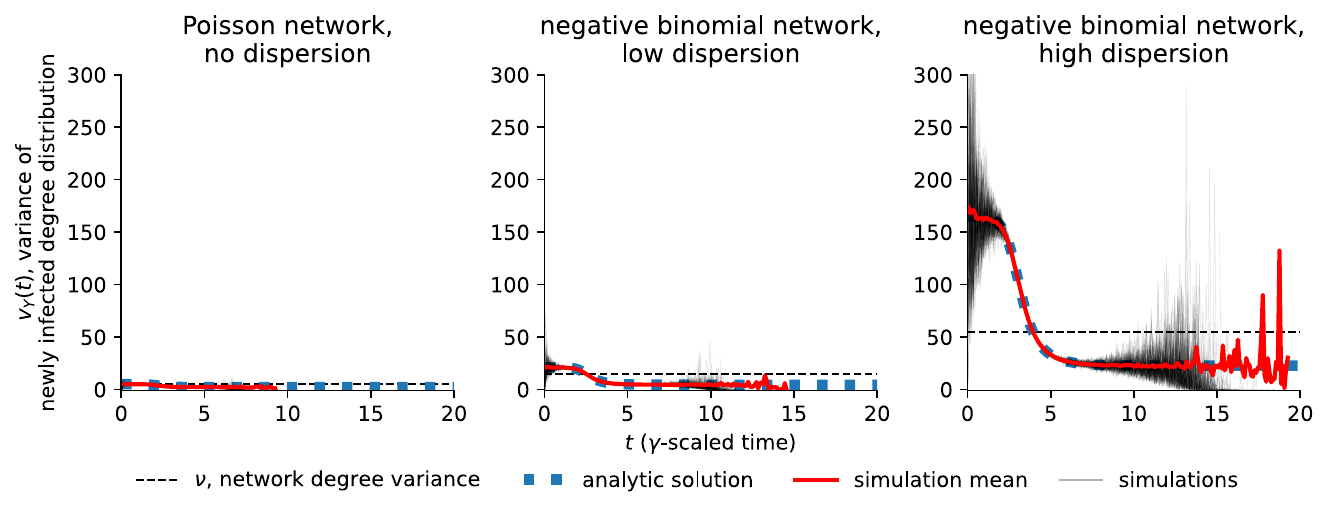}
\caption{Trajectories of the newly infected degree distribution's variance $v_Y(t)$ for three different network degree distributions. Red curves show the mean of 200 simulations (whose individual trajectories are also plotted by faint gray curves) while the blue dotted curves show the analytic solution from Eqs.~\eqref{eq:supplement:mY}--\eqref{eq:supplement:mY2}. Black dashed lines show the variance $\nu$ of the network degree distribution.}
\label{fig:vYs}
\end{figure}

\clearpage

\begin{theorem}
The mass function $p_{F,k}(t)$ of the effective newly infected degree distribution $F(t)$ is given by
\begin{equation}
\label{eq:supplement:pFk}
p_{F,k}=\sum_{j=k+1}^\infty\frac{P(j)j\theta^{j-1}}{\psi'(\theta)}\binom{j-1}{k}\left(\frac1\theta\frac{\psi'(\theta)}{\psi'(1)}\right)^k\left(1-\frac1\theta\frac{\psi'(\theta)}{\psi'(1)}\right)^{j-1-k}
\end{equation}
and, for any $n$ for which $\mu_{n+1}$ exists, its $n$-th moment $m_{F,n}(t)$ is given by
\begin{equation}
\label{eq:supplement:mFn}
m_{F,n}=\sum_{i=0}^n\stirling ni\frac{\psi'(\theta)^{i-1}}{\psi'(1)^i}\psi^{(i+1)}(\theta)
\end{equation}
where $\stirling ni$ are Stirling numbers of the second kind. In particular, $F(t)$ has mean
\begin{equation}
\label{eq:supplement:mF}
m_F=m_{F,1}=\frac{\psi''(\theta)}{\psi'(1)}
\end{equation}
and second moment
\begin{equation}
\label{eq:supplement:mF2}
m_{F,2}=\frac{\psi''(\theta)}{\psi'(1)}+\frac{\psi'''(\theta)\psi'(\theta)}{\psi'(1)^2}.
\end{equation}
\end{theorem}
\begin{proof}
Let $x$ be an arbitrary node, and let $y$ be an arbitrary neighbor of $x$. We now calculate the following probability using Bayes' law:
\begin{align}
&\P\big(y\text{ is susceptible at time }t\mid y\text{ never infected }x\text{ by time }t\big)\nonumber\\
&=\frac{\P\big(y\text{ never infected }x\text{ by time }t\mid y\text{ is susceptible at time }t\big)\ \P\big(y\text{ is susceptible at time }t\big)}{\P\big(y\text{ never infected }x\text{ by time }t\big)}\nonumber\\
&=\frac1{\theta(t)}\frac{\psi'(\theta(t))}{\psi'(1)}, \label{eq:supplement:pYk_calc}
\end{align}
where $\P\big(y\text{ never infected }x\text{ by time }t\mid y\text{ is susceptible at time }t\big)=1$ since $y$ cannot have infected $x$ by time $t$ if $y$ had been susceptible the whole time, and $\P\big(y\text{ is susceptible at time }t)=\sum\limits_{k=1}^\infty\frac{P(k)k\theta(t)^{k-1}}{\psi'(1)}$ since $y$ is a random neighbor of $x$.

Now let $x$ be a node newly infected at time $t$, so that it's degree $j$ follows the newly infected degree distribution $Y(t)$. To calculate the probability that an arbitrary neighbor $y$ of $x$ (other than the neighbor which infected $x$) is susceptible at time $t$, we must use the additional information that $y$ had never infected $x$ by time $t$. Therefore, the probability that $y$ is susceptible is just that given by Eq.~\eqref{eq:supplement:pYk_calc}, namely $\frac1{\theta}\frac{\psi'(\theta)}{\psi'(1)}$. And since $x$ can have at most $j-1$ susceptible neighbors, the probability that $x$ has exactly $k\le j-1$ susceptible neighbors is given by the binomial expression $\binom{j-1}k\left(\frac1{\theta}\frac{\psi'(\theta)}{\psi'(1)}\right)^k\left(1-\frac1{\theta}\frac{\psi'(\theta)}{\psi'(1)}\right)^{j-1-k}$. Thus,
\begin{equation*}
p_{F,k}=\sum_{j=k+1}^\infty p_{Y,k}\binom{j-1}{k}\left(\frac1\theta\frac{\psi'(\theta)}{\psi'(1)}\right)^k\left(1-\frac1\theta\frac{\psi'(\theta)}{\psi'(1)}\right)^{j-1-k},
\end{equation*}
into which we plug in Eq.~\eqref{eq:supplement:pYk} to get the result Eq.~\eqref{eq:supplement:pFk}.

The moments of $F(t)$ follow from
\begin{align*}
m_{F,n}&=\sum_{k=0}^np_{F,k}k^n\\
&=\sum_{k=0}^nk^n\sum_{j=k+1}^\infty\frac{P(j)j\theta^{j-1}}{\psi'(\theta)}\binom{j-1}{k}\left(\frac1\theta\frac{\psi'(\theta)}{\psi'(1)}\right)^k\left(1-\frac1\theta\frac{\psi'(\theta)}{\psi'(1)}\right)^{j-1-k}\\
&=\sum_{j=1}^\infty\frac{P(j)j\theta^{j-1}}{\psi'(\theta)}\sum_{k=0}^{j-1}k^n\binom{j-1}{k}\left(\frac1\theta\frac{\psi'(\theta)}{\psi'(1)}\right)^k\left(1-\frac1\theta\frac{\psi'(\theta)}{\psi'(1)}\right)^{j-1-k}\\
&=\sum_{j=1}^\infty\frac{P(j)j\theta^{j-1}}{\psi'(\theta)}\sum_{i=0}^n\stirling ni(j-1)^{\underline{i}}\left(\frac1\theta\frac{\psi'(\theta)}{\psi'(1)}\right)^i\\
&=\sum_{i=0}^n\stirling ni\frac{\psi'(\theta)^{i-1}}{\psi'(1)^i}\sum_{j=1}^\infty P(j)j^{\underline{i+1}}\theta^{j-(i+1)},
\end{align*}
where we applied the formula for the $n$-th moment of a binomial distribution and $j^{\underline{k}}=j(j-1)\cdots(j-k+1)$ is a falling factorial. We then get the result Eq.~\eqref{eq:supplement:mFn} by noting that $\psi^{i+1}(\theta)=\sum\limits_{j=1}^\infty P(j)j^{\underline{i+1}}\theta^{j-(i+1)}$, which will converge for any $i\le n$ by the existence of $\mu_{n+1}$, since $\psi^{(i+1)}(\theta)\le\langle K^{i+1}\theta^K\rangle\le\langle K^{n+1}\rangle=\mu_{n+1}$.
\end{proof}

The mean and variance of the effective newly infected degree distribution $F(t)$ are shown in Figs.~\ref{fig:mFs}--\ref{fig:vFs}, respectively.

\clearpage

\begin{figure}[!t]
\centering
\includegraphics[width=\textwidth]{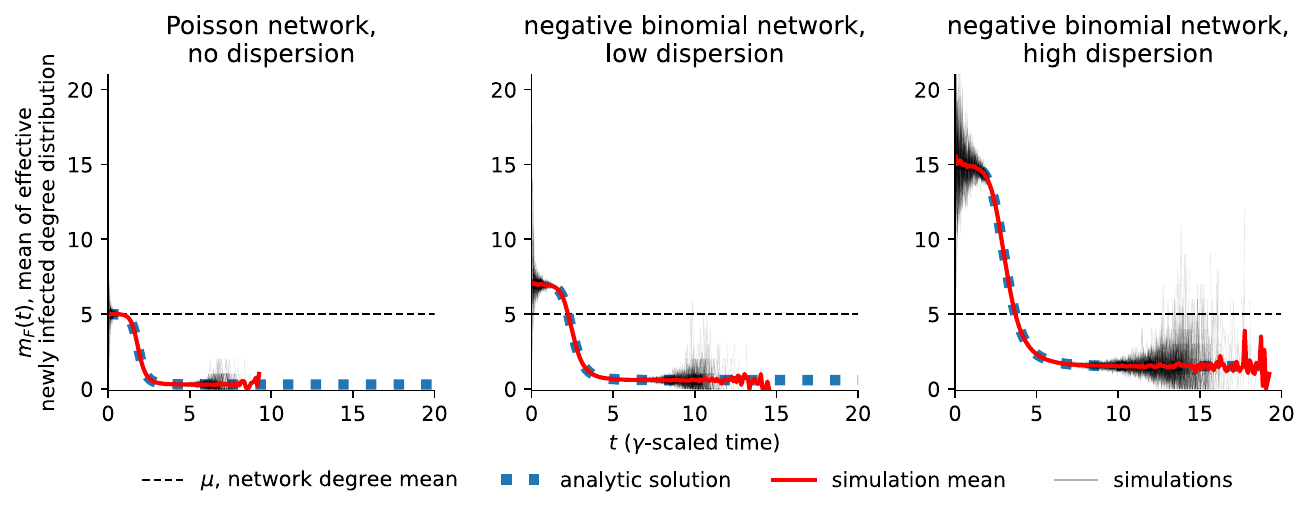}
\caption{Trajectories of the effective newly infected degree distribution's mean $m_F(t)$ for three different network degree distributions. Red curves show the mean of 200 simulations (whose individual trajectories are also plotted by faint gray curves) while the blue dotted curves show the analytic solution from Eq.~\eqref{eq:supplement:mF}. Black dashed lines show the mean $\mu$ of the network degree distribution.}
\label{fig:mFs}
\end{figure}

\begin{figure}[!b]
\centering
\includegraphics[width=\textwidth]{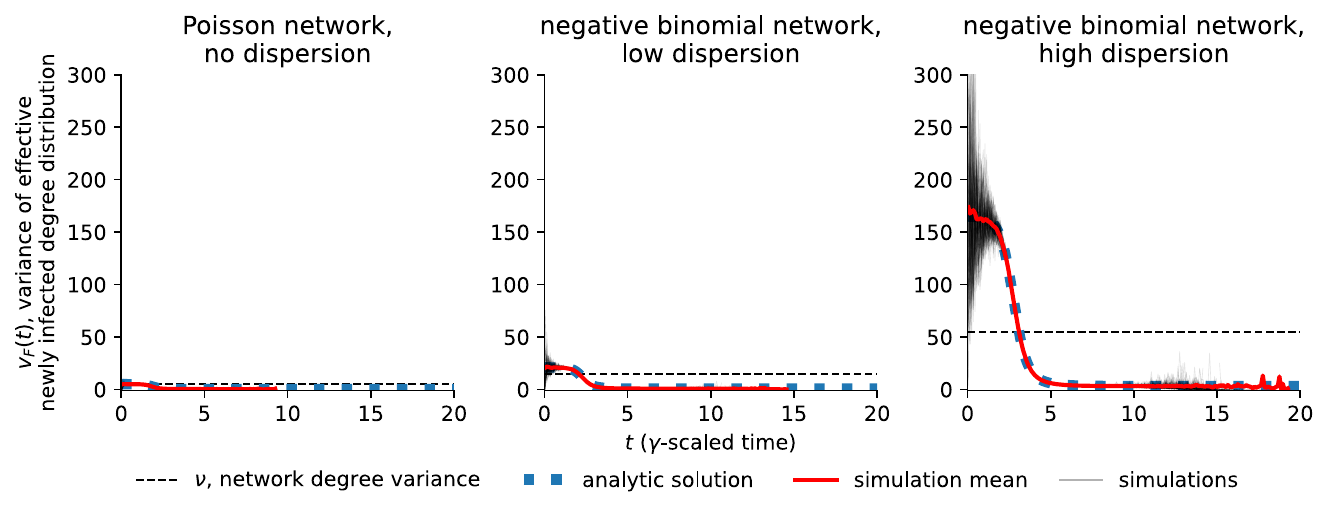}
\caption{Trajectories of the effective newly infected degree distribution's variance $v_F(t)$ for three different network degree distributions. Red curves show the mean of 200 simulations (whose individual trajectories are also plotted by faint gray curves) while the blue dotted curves show the analytic solution from Eqs.~\eqref{eq:supplement:mF}--\eqref{eq:supplement:mF2}. Black dashed lines show the variance $\nu$ of the network degree distribution.}
\label{fig:vFs}
\end{figure}

\clearpage
\section{Proofs from the main text in more detail}

\begin{customlemma}{2 from the main text}
The following results are true for the infection turnover rate $\frac{J(t)}{I(t)}$ when $\theta(0)$ is sufficiently close to 1:
\begin{enumerate}[(1)]
\item $\frac{J(t)}{I(t)}>\gamma$ while $\dot J(t)\ge0$.

\item $\frac{J(t)}{I(t)}<\beta\left(\frac{\psi''(1)}{\psi'(1)}-1\right)$ always, and $\frac{J(t)}{I(t)}$ can be made to stay arbitrarily close to $\beta\left(\frac{\psi''(1)}{\psi'(1)}-1\right)$ for an arbitrary amount of time while $\theta(t)\approx1$ by choosing $\theta(0)$ sufficiently close to 1.

\item $\lim\limits_{t\to\infty}\frac{J(t)}{I(t)}=\max\left\{\beta\left(\frac{\psi''(\theta(\infty))}{\psi'(1)}-1\right),0\right\}$.
\end{enumerate}
\end{customlemma}
\begin{proof}
We begin by establishing that $J(t)$ and $I(t)$ have peaks at times greater than 0. For $J(t)$, we see from its logarithmic growth rate
\begin{equation}
\label{eq:supplement:dJ}
\frac d{dt}\log J=
\frac{\dot J}{J}=\frac{\psi''(\theta)}{\psi'(\theta)}\dot\theta+\beta\frac{\psi''(\theta)}{\psi'(1)}-(\beta+\gamma)
\end{equation}
that $J(t)$ is increasing at $t=0$ (since $\beta\frac{\psi''(\theta(0))}{\psi'(1)}-(\beta+\gamma)\approx(\beta+\gamma)(\RO-1)>0$ and $\dot\theta(0)\approx0$) and that $J(t)$ is decreasing for $t\to\infty$ (since $\psi''(\theta(\infty))<\frac{\psi'(1)-\psi'(\theta(\infty))}{1-\theta(\infty)}=\frac{\beta+\gamma}\beta\psi'(1)$ by the convexity of $\psi'$ and $\dot\theta<0$). Thus, $J(t)$ has a peak time $t_J>0$. And since the peak $t_J>0$ of $J(t)$ exists, it is easy to see that the peak $t_I>0$ of $I(t)$ will also exist from the fact that $\dot I=J-\gamma I$ and that $I$ initially increases (as shown for $\RO>1$ by Miller \cite{Miller2011}).

Now to prove the first result, note that $\dot J=\ddot I+\gamma\dot I$. At the peak of $I(t)$, $\dot I=0$ and $\ddot I<0$, implying $\dot J<0$. It must then be true that at the peak of incidence $J(t)$, the prevalence $I(t)$ has not yet reached its peak and so $\dot I>0$. Therefore, at and before the peak of incidence, we must have $\frac JI=\frac{\dot I+\gamma I}I>\gamma$.

For the second result, define $\alpha(t)=\beta\left(\frac{\psi''(1)}{\psi'(1)}-1\right)I(t)-J(t)$. First we see that $\alpha(0)>0$, since for $\theta(0)\approx1$ we have $I(0)=1-\psi(\theta(0))\approx\psi'(1)(1-\theta(0))$ and
\begin{equation*}
\begin{aligned}
J(0)&=\psi'(\theta(0))\left(\beta\frac{\psi'(1)-\psi'(\theta(0))}{\psi'(1)}-(\beta+\gamma)(1-\theta(0))\right)\\
&\approx\big(\beta\psi''(1)-(\beta+\gamma)\psi'(1)\big)(1-\theta(0)),
\end{aligned}
\end{equation*}
so that
\begin{equation}
\lim_{\theta(0)\to1}\frac{J(0)}{I(0)}=\beta\frac{\psi''(1)}{\psi'(1)}-(\beta+\gamma).
\label{eq:supplement:JI0}
\end{equation}
Now we can easily derive from $\dot J=-\psi''(\theta)\dot\theta^2+\beta(\frac{\psi''(\theta)}{\psi'(1)}-1)J-\gamma J$ that
\begin{equation}
\dot\alpha=\psi''(\theta)\dot\theta^2+\beta\frac{\psi''(1)-\psi''(\theta)}{\psi'(1)}J-\gamma\alpha,
\label{eq:supplement:dalpha}
\end{equation}
which, in combination with $\alpha(0)>0$, shows that $\alpha(t)>0$ for all $t$. This proves the first part of the second result, that $\frac{J(t)}{I(t)}<\beta\left(\frac{\psi''(1)}{\psi'(1)}-1\right)$ always, assuming $\theta(0)$ is sufficiently close to 1.

The second part of the second result also follows from Eq.~\eqref{eq:supplement:dalpha}. To see this, note that the first two terms of Eq.~\eqref{eq:supplement:dalpha} both go to 0 as $\theta\to1$, so while $\theta(t)\approx1$, $\alpha(t)$ tends at rate $\gamma$ to very small quantities (which can be made arbitrarily small for as long as desired by choosing $\theta(0)$ sufficiently close to 1). And since $I(t)$ is also increasing while $\theta(1)\approx1$, then $\beta\left(\frac{\psi''(1)}{\psi'(1)}-1\right)-\frac{J(t)}{I(t)}=\frac{\alpha(t)}{I(t)}$ can be made arbitrarily small for arbitrarily long by choosing $\theta(0)$ sufficiently close to 1.

For the third result, we can formulate $I(t)$ and $J(t)$ explicitly as
\begin{equation}
I(t)=e^{-\gamma t}I(0)+\int_0^te^{-\gamma(t-s)}J(s)ds
\label{eq:supplement:I_integral}
\end{equation}
and
\begin{equation}
J(t)=-\psi(\theta(t))\dot\theta(0)\exp\left(\int_0^t\left(\beta\frac{\psi''(\theta(s))}{\psi'(1)}-(\beta+\gamma)\right)ds\right)
\label{eq:supplement:J_integral}
\end{equation}
from $\frac{\ddot\theta}{\dot\theta}=\beta\frac{\psi''(\theta)}{\psi'(1)}-(\beta+\gamma).$ Thus, we can get
\begin{equation}
\frac{J(t)}{I(t)} = \frac{-\psi(\theta(t))\dot\theta(0)\exp\left(\int_0^t\beta\left(\frac{\psi''(\theta(s))}{\psi'(1)}-1\right)ds\right)}{I(0)-\int_0^t\psi(\theta(s))\dot\theta(0)\exp\left(\int_0^s\beta\left(\frac{\psi''(\theta(\tau))}{\psi'(1)}-1\right)d\tau\right)ds}.
\label{eq:supplement:JI_integral}
\end{equation}
Now let $c=\beta\left(\frac{\psi''(\theta(\infty))}{\psi'(1)}-1\right)$. We consider the behavior as $t\to\infty$ of Eq.~\eqref{eq:supplement:JI_integral} for the cases of $c>0$, $c<0$, and $c=0$. In the first case ($c>0$),
\begin{equation}
\frac{J(t)}{I(t)}\to\frac{-\psi(\theta(\infty))\dot\theta(0)\frac1{c}\left(e^{ct}-1\right)}{-\int_0^t\psi(\theta(\infty))\dot\theta(0)\frac1{c}\left(e^{cs}-1\right)ds}\to c\text{ as }t\to\infty,
\label{eq:supplement:JI_infty_case_1}
\end{equation}
as the term encompassed by the outer integral in the denominator of Eq.~\eqref{eq:supplement:JI_integral} for later $t$ will come to dominate that term in the integral for earlier $t$ as well as the constant $I(0)$. In the second case ($c<0$), the numerator of Eq.~\eqref{eq:supplement:JI_integral} goes to 0 as $t\to\infty$ while the denominator remains at least $I(0)$, and so $\lim\limits_{t\to0}\frac{J(t)}{I(t)}=0$. In the last case ($c=0$), we first must note that the integral $\int_0^\infty\beta\left(\frac{\psi''(\theta(s))}{\psi'(1)}-1\right)ds$ will converge since $\theta(s)-\theta(\infty)$ for large $s$ behaves as an exponential decay (see Eq.~\eqref{eq:supplement:dTheta} below). Thus, for large $t$, Eq.~\eqref{eq:supplement:JI_integral} will look like $\frac{A}{B+At}$ with $A=-\psi(\theta(\infty))\dot\theta(0)\exp\left(\int_0^\infty\beta\left(\frac{\psi''(\theta(s))}{\psi'(1)}-1\right)ds\right)$ and $B$ some positive constant, showing that $\lim\limits_{t\to0}\frac{J(t)}{I(t)}=0.$ To complete the proof just requires establishing $\theta(t)-\theta(\infty)$ behaves as an exponential decay for large $t$. Let $\Theta(t)=\theta(t)-\theta(\infty)$ and then this follows from the following approximation for when $\Theta\approx0$:
\begin{equation}
\begin{aligned}
\dot\Theta&\approx-\beta(\Theta+\theta(\infty))+\beta\frac{\psi'(\theta(\infty))+\psi''(\theta(\infty))\Theta}{\psi'(1)}+\gamma\big(1-(\Theta+\theta(\infty))\big)\\
&=\left(\beta\frac{\psi''(\theta(\infty))}{\psi'(1)}-(\beta+\gamma)\right)\Theta,
\end{aligned}
\label{eq:supplement:dTheta}
\end{equation}
where the first approximation follows from the Taylor expansion $\psi'(\theta(\infty)+\Theta)\approx\psi'(\theta(\infty))+\psi''(\theta(\infty))\Theta$ and the second step follows from $\dot\theta(\infty)=-\beta\theta(\infty)+\beta\frac{\psi'(\theta(\infty))}{\psi'(1)}+\gamma(1-\theta(\infty))=0$. We remind the reader that $\beta\frac{\psi''(\theta(\infty))}{\psi'(1)}-(\beta+\gamma)<0$, as we showed at the beginning of this proof, to ensure that the behavior of $\Theta(t)$ is in fact an exponential decay for large $t$.
\end{proof}

\begin{customtheorem}{3 from the main text}
In the limit as $\theta(0)\approx 1$, the infected degree distribution $X(t)$ will approach the neighbor degree distribution $K_{\textnormal{n}}$ at some time $t$. Specifically, for any small $\epsilon>0$ and $k\in\N$, there exists a $\theta(0)$ and time $t_\epsilon$ for which $\abs{p_j(t_\epsilon)-\frac{P(k)k}{\mu}}<\epsilon$ for all $j\le k$.
\label{thm:supplement:X_KN}
\end{customtheorem}
\begin{proof}
Filling in the gaps from the proof outline in the main text, we first show that the peak time $t_J$ of $J(t)$ is greater than 0 and that $\theta_J=\theta(t_J)$ is constant regardless of initial condition. The fact that $t_J>0$ we established in our proof here of Lemma 2. To see that $\theta_J$ is constant regardless of initial conditions, it suffices to set $\dot J(t)=\frac d{dt}\left(-\psi'(\theta(t))\dot\theta(t)\right)=0$, which gives $\psi''(\theta)\dot\theta^2+\psi'(\theta)\ddot\theta=0$ and thus
\begin{equation}
\frac{\psi''(\theta_J)}{\psi'(\theta_J)}\left(-\beta\theta_J+\beta\frac{\psi'(\theta_J)}{\psi'(1)}+\gamma(1-\theta_J)\right)+\beta\frac{\psi''(\theta_J)}{\psi'(1)}-(\beta+\gamma)=0,
\label{eq:supplement:thetaJ}
\end{equation}
from which we see that $\theta_J$ is evidently independent of initial conditions.

We now give an explicit choice of $\theta(0)$ and $t_\epsilon$ which ensure, for a given choice of small $\epsilon>0$ and large $k\in\N$, that $\abs{p_j(t_\epsilon)-\frac{P(k)k}{\mu}}<\epsilon$ for all $j\le k$. First, we repeat the following two equations from the main text for convenience:
\begin{equation}
\dot p_k=-\frac{J}{I}\left(p_k-\frac{P(k)k\theta^k}{\phi'(\log\theta)}\right)
\label{eq:supplement:dpk}
\end{equation}
and
\begin{equation}
\label{eq:supplement:theta_ineq}
\theta(t)>1-(1-\theta(0))e^{(\beta+\gamma)(\RO-1)t}\ \text{ for all }t>0.
\end{equation}
Now let $\theta_\epsilon$ be a value of $\theta$ close enough to 1 so that
\begin{equation}
\label{eq:supplement:theta_epsilon}
\abs{\frac{P(j)j\theta^j}{\phi'(\log\theta)}-\frac{P(j)j}{\mu}}<\frac\epsilon2\text{ for all }\theta_\epsilon\le\theta\le1\text{ and }j\le k.
\end{equation}
Then choose $t_\epsilon>0$ such that
\begin{align*}
\abs{\left(P(j)-\frac{P(j)j}{\mu}\right)e^{-\gamma t_\epsilon}}<\frac\epsilon2\text{ for all }j\le k.
\end{align*}
We define $t_\epsilon$ this way since Eq.~\eqref{eq:supplement:dpk}, Lemma 2, and $p_j(0)=P(j)$ imply that $p_j(t)$ will approach $\frac{P(j)j}{\mu}$ faster than will $\frac{P(j)j}{\mu}+\big(P(j)-\frac{P(j)j}{\mu}\big)e^{-\gamma t}$ while $\theta(t)\approx1$. Of course, $p_j(t)$ is not actually approaching $\frac{P(j)j}{\mu}$, but while $\theta(t)\ge\theta_\epsilon$, it is approaching something that is within $\frac{\epsilon}2$ of $\frac{P(j)j}{\mu}$; thus, if $\theta(t)\ge\theta_\epsilon$ for all $t\le t_\epsilon$, we will have $\abs{p_j(t_\epsilon)-\frac{P(k)k}{\mu}}<\epsilon$ for all $j\le k$. In order to ensure $\theta(t)\ge\theta_\epsilon$ for all $t\le t_\epsilon$, we simply let $\theta(0)=1-\frac{1-\theta_\epsilon}{e^{(\beta+\gamma)(\RO-1)t_\epsilon}}$ and then we will have our desired result by Eq.~\eqref{eq:supplement:theta_ineq}.
\end{proof}

\begin{customtheorem}{5 from the main text}
\label{thm:supplement:tm}
If $\theta(0)$ is sufficiently close to 1, then $m(t)$ and $I(t)$ will both peak at times $t_m>0$ and $t_I>0$ respectively. And as $\theta(0)$ approaches 1 both $t_m$ and $t_I$ will diverge to $\infty$ while
\begin{equation}
0<\lim_{\theta(0)\to1}\frac{t_m}{t_I}<\frac12-\frac{\gamma}{4(\beta+\gamma)(\RO-1)+2\gamma}.
\label{eq:supplement:tm}
\end{equation}
\end{customtheorem}
\begin{proof}
Filling in the gaps from the proof outline in the main text, we first show why $\lim\limits_{\theta(0)\to\infty}t_m=\infty$. From Eq.~\eqref{eq:supplement:theta_ineq}, we see that $\theta(t)$ can be kept arbitrarily close to 1, and thus $\frac{\phi''(\log\theta(t))}{\phi'(\log\theta(t))}$ kept arbitrarily close to $\frac{\phi''(0)}{\phi'(0)}$, for as long as desired by choosing $\theta(0)$ sufficiently close to 1. Then since $\dot m=-\frac{J}{I}\left(m-\frac{\phi''(\log\theta)}{\phi'(\log\theta)}\right)$, $m(t)$ can be made to take arbitrarily long to peak by choosing $\theta(0)$ sufficiently close to 1 so that the moving target $\frac{\phi''(\log\theta)}{\phi'(\log\theta)}$ that $m(t)$ tends toward stays sufficiently constant, noting that the rate at which $m(t)$ tends towards this target ($\frac{J(t)}{I(t)}$) is bounded above and below by positive constants by Lemma 2. In order to apply Lemma 2 here to get that $\frac{J(t)}{I(t)}$ is bounded below by $\gamma$, it must be noted that $\theta(t)\ge\theta_J$ (we showed previously that $\theta_J$ is a constant regardless of initial condition that satisfies Eq.~\eqref{eq:supplement:thetaJ}) and thus $t\le t_J$ while $\theta(t)$ is kept sufficiently close to 1, as we are assuming is the case here.

Next we prove that $\lim\limits_{\theta(0)\to1}\frac{t_I}{t_J}=1$. We first point out that in our proof of the first result of Lemma 2, we showed that $J(t)$ peaks before $I(t)$, or $t_J<t_I$. In the proof of Lemma 2, we also use Eq.~\eqref{eq:supplement:I_integral}, which shows that as $\theta(0)\to1$ and $t_J\to\infty$, the behavior of $I(t)$ for after $t_J$ will approximate
\begin{equation}
\label{eq:supplement:I_after_tJ}
\begin{aligned}
&I(t_J+\tau)\approx C+\\
&\qquad\int_{0}^{\tau}e^{-\gamma(\tau-s)}\psi'(\theta(t_J+s))\left(\beta\theta(t_J+s)-\beta\frac{\psi'(\theta(t_J+s))}{\psi'(1)}-\gamma(1-\theta(t_J+s))\right)ds
\end{aligned}
\end{equation}
for $\tau>0$, where $C=\int_0^{t_J} e^{-\gamma(t-s)}\psi'(\theta(s))\left(\beta\theta(s)-\beta\frac{\psi'(\theta(s))}{\psi'(1)}-\gamma(1-\theta(s))\right)$ approaches a constant as $\theta(0)\to1$ and $t_J\to\infty$ due to the deterministic trajectory of $\theta$. Thus, the time $\tau$ that it takes for $I(t)$ to peak after $J(t)$ peaks, or $t_I-t_J$, will approach a constant as $\theta(0)\to 1$. Then since both $t_I$ and $t_J$ go to $\infty$ as $\theta(0)\to 1$, their ratio $\frac{t_I}{t_J}$ will go to 1.

Next we show $\lambda_2=(\beta+\gamma)(\RO-1)-\frac\beta{\psi'(1)}\left(\psi''(1)-\frac{\psi'(1)-\psi'(\theta_J)}{1-\theta_J}\right)>0$. Rearranging this inequality gives
\begin{equation}
\label{eq:supplement:tJ_inequality}
1-\theta_J<\frac\beta{\beta+\gamma}\frac{\psi'(1)-\psi'(\theta_J)}{\psi'(1)},
\end{equation}
so this is the result we are trying to prove. Now rearranging Eq.~\eqref{eq:supplement:thetaJ}, we get
\begin{align*}
1-\theta_J&=\frac\beta{\beta+\gamma}+\frac{\psi'(\theta_J)}{\psi''(\theta_J)}-2\frac\beta{\beta+\gamma}\frac{\psi'(\theta_J)}{\psi'(1)}\\
&=\frac\beta{\beta+\gamma}\frac{\psi'(1)-\psi'(\theta_J)}{\psi'(1)}-\frac{\psi'(\theta_J)}{(\beta+\gamma)\psi''(\theta_J)}\left(\beta\frac{\psi''(\theta_J)}{\psi'(1)}-(\beta+\gamma)\right),
\end{align*}
so it just suffices to show that $\beta\frac{\psi''(\theta_J)}{\psi'(1)}-(\beta+\gamma)>0$. This follows from Eq.~\eqref{eq:supplement:thetaJ} because the expression inside the large parentheses there is $\dot\theta(t_J)$, which is negative by the ever-decreasing nature of $\theta(t).$

Lastly, we give an explicit formula for $t_m$  in the limit as $\theta(0)\to1$, derived from (27) in the main text, which we reproduce here for convenience:
\begin{equation}
\lim\limits_{\theta(0)\to1}\dot m(t)=-\frac{J(t)}{I(t)}\left(m-\big(f(1)-f'(1)(1-\theta(0))e^{(\beta+\gamma)(\RO-1)t}\big)\right)\text{ for all }t\le t_m
\label{eq:supplement:dm_limit}
\end{equation}
converges uniformly across $t\le t_m$, with $f(x)=\frac{\phi''(\log x)}{\phi'(\log x)}$. Thus $f(1)=\frac{\phi''(0)}{\phi'(0)}$ and $f'(1)=\frac{\phi'''(0)}{\phi'(0)}-\left(\frac{\phi''(0)}{\phi'(0)}\right)^2$. We now treat $\dot m(t)$ as being (approximately) equal to its limit from Eq.~\eqref{eq:supplement:dm_limit} and treat $\frac{J(t)}{I(t)}$ as being (approximately) equal to $\beta\left(\frac{\psi''(1)}{\psi'(1)}-1\right)$, since Lemma 2 shows that $\frac{J(t)}{I(t)}$ can be arbitrarily close to $\beta\left(\frac{\psi''(1)}{\psi'(1)}-1\right)$ for as long as desired while $\theta(t)\approx1$ by choosing $\theta(0)$ sufficiently close to 1. Doing so gives the (approximate) differential equation
\begin{equation}
\dot m(t)\approx-(\lambda_1+\gamma)\left(m-\big(f(1)-f'(1)(1-\theta(0))e^{\lambda_1 t}\big)\right)
\label{eq:supplement:dm_approx}
\end{equation}
with $\lambda_1=(\beta+\gamma)(\RO-1)$ and $m(0)=\mu$, which has explicit solution
\begin{equation}
m(t)\approx f(1)-\frac{\lambda_1+\gamma}{2\lambda_1+\gamma}f'(1)(1-\theta(0))e^{\lambda_1t}-Ce^{-(\lambda_1+\gamma)t}
\end{equation}
with $C=f(1)-\mu-\frac{\lambda_1+\gamma}{2\lambda_1+\gamma}f'(1)(1-\theta(0))
$ a constant so that $m(0)=\mu$. To find the peak time $t_m$ in this limit, we set Eq.~\eqref{eq:supplement:dm_approx} equal to 0, or equivalently, set 
\begin{equation*}
f(1)-\frac{\lambda_1+\gamma}{2\lambda_1+\gamma}f'(1)(1-\theta(0))e^{\lambda_1t_m}-Ce^{-(\lambda_1+\gamma)t_m}\approx f(1)-f'(1)(1-\theta(0))e^{\lambda_1 t_m}
\end{equation*}
which gives
\begin{equation}
t_m\approx\frac1{2\lambda_1+\gamma}\left(\log\frac{1}{1-\theta(0)}+\log\frac{(2+\frac\gamma{\lambda_1})C}{f'(1)}\right).
\end{equation}
Again this approximation becomes asymptotically correct as $\theta(0)$ approaches 1. As $\theta(0)\to1$, the $\log\frac{1}{1-\theta(0)}$ term will dominate and we get $\lim\limits_{\theta(0)\to1}\frac{t_m}{\log\frac1{1-\theta(0)}}=\frac1{2\lambda_1+\gamma}$ as desired.
\end{proof}

\end{document}